\documentclass[12pt, a4paper]{amsart}
\usepackage[dvips]{epsfig}
\usepackage{amsmath}
\usepackage{amssymb}
\usepackage{amsfonts}
\usepackage{bbm}
\usepackage{amsthm}
\usepackage{amsbsy}
\usepackage{amsgen}
\usepackage{amscd}
\usepackage{amsopn}
\usepackage{amstext}
\usepackage{amsxtra}
\usepackage[utf8]{inputenc}
\usepackage{enumerate}
\usepackage{hyperref}
\usepackage{lineno}
\usepackage{marginnote}
\usepackage{verbatim}
\usepackage{calrsfs}
\usepackage{times, graphicx}
\usepackage{layout}
\usepackage{subcaption}
\usepackage{eso-pic}
\usepackage{lastpage}
\DeclareMathAlphabet{\pazocal}{OMS}{zplm}{m}{n}
\usepackage[foot]{amsaddr}

  \evensidemargin
\oddsidemargin
\numberwithin{equation}{section}

\makeatletter
\newcommand*\rel@kern[1]{\kern#1\dimexpr\macc@kerna}
\newcommand*\widebar[1]{%
  \begingroup
  \def\mathaccent##1##2{%
    \rel@kern{0.8}%
    \overline{\rel@kern{-0.8}\macc@nucleus\rel@kern{0.2}}%
    \rel@kern{-0.2}%
  }%
  \macc@depth\@ne
  \let\math@bgroup\@empty \let\math@egroup\macc@set@skewchar
  \mathsurround\z@ \frozen@everymath{\mathgroup\macc@group\relax}%
  \macc@set@skewchar\relax
  \let\mathaccentV\macc@nested@a
  \macc@nested@a\relax111{#1}%
  \endgroup
}
\makeatother

\newtheorem{theorem}{Theorem}[section]
\newtheorem{proposition}[theorem]{Proposition}
\newtheorem{lemma}[theorem]{Lemma}
\newtheorem{corollary}[theorem]{Corollary}

\theoremstyle{definition}
\newtheorem{definition}[theorem]{Definition}

\newcommand{\vol}{\operatorname{Vol}}
\newcommand{\area}{\operatorname{Area}}
\newcommand{\prob}{\pazocal{P}r}

\newcommand{\Dpc}{\pazocal{D}}

\newcommand{\Sc}{\pazocal{S}}

\newcommand{\Ec}{\pazocal{E}}

\newcommand{\Bc}{\mathcal{B}}

\newcommand{\Fc}{\pazocal{F}}

\newcommand{\Q}{\mathbb{Q}}

\newcommand{\meas}{\operatorname{meas}}

\newcommand {\E} {\mathbb{E}}
\newcommand {\T} {\mathbb{T}}

\newcommand {\R} {\mathbb{R}}

\newcommand {\Nb} {\mathbb{N}}
\newcommand {\Z} {\mathbb{Z}}

\newcommand {\Cc} {\pazocal{C}}
\newcommand {\Lc} {\mathcal{L}}

\newcommand {\Pc} {\mathcal{P}}

\newcommand {\C} {\mathbb{C}}

\newcommand {\Ac} {\pazocal{A}}

\newcommand {\Tb} {\mathbb{T}}
\newcommand {\var} {\operatorname{Var}}

\newcommand{\hidefixme}[1]{}

 \allowbreak

\begin{document}

\title[Toral defect]{The defect of toral Laplace eigenfunctions and Arithmetic Random Waves}

\author{P\"{a}r Kurlberg\textsuperscript{1}}
\email{kurlberg@kth.se}
\author{Igor Wigman\textsuperscript{2}}
\email{igor.wigman@kcl.ac.uk}
\author{Nadav Yesha\textsuperscript{3}}
\email{nyesha@univ.haifa.ac.il}

\date{\today}

\begin{abstract}

We study the defect (or ``signed area") distribution of toral Laplace eigenfunctions restricted to
shrinking balls of radius above the Planck scale, in either random Gaussian scenario (``Arithmetic
Random Waves"), or deterministic eigenfunctions averaged w.r.t. the spatial variable.
In either scenario we exploit the associated symmetry of the eigenfunctions
to show that the expectation (Gaussian or spatial) vanishes. Our principal results concern the high energy
limit behaviour of the defect variance.

\end{abstract}
	
\maketitle

\section{Introduction}
\subsection{Toral Laplace eigenfunctions and Arithmetic Random Waves}

Toral Laplace eigenfunctions are an important model in Quantum Chaos
that represent the Laplace eigenfunctions on generic
manifolds. From the point of view
of an investigator interested in the study of their properties, the
toral eigenfunctions enjoy two significant privileges over the general
case, making them attractive to address, in addition to their own
sake, being Fourier sums with particular frequencies. First, its
number theoretic ingredient makes them susceptible to methods borrowed
from Analytic Number Theory. Second, their (slowly in $2$
dimensions) growing spectral degeneracies allow for the study of the
``typical" case, whether that means endowing the linear space of
Laplace eigenfunctions with the same eigenvalue with a Gaussian
probability measure (thus giving rise to ``Arithmetic Random Waves"),
or otherwise.

Let $\Tb^{2} = \R^{2}/\Z^{2}$ be the standard $2$-torus, $$S=\{a^{2}+b^{2}:\: a,b\in\Z\}$$ be the set of
all integers expressible as sum of two squares (``sequence of toral energies"),
and for $n\in S$ let $$N_{n}:=r_{2}(n)=\#\left\{(a,b)\in\Z^{2}:\: n=a^{2}+b^{2}\right\}$$
be the number of ways to express $n$ as sum of two squares. Then every function of the form
\begin{equation}
\label{eq:fn toral Laplace eig}
f_{n}(x) = \frac{1}{\sqrt{2N_{n}}}\sum\limits_{\lambda\in\Z^{2}:\: \|\lambda\|^{2}=n} a_{\lambda}\cdot e(\langle x,\lambda\rangle)
\end{equation}
with convenience only pre-factor $\frac{1}{\sqrt{2N_{n}}}$, $\lambda=(\lambda_{1},\lambda_{2})\in\Z^{2}$, $x=(x_{1},x_{2})\in \Tb^{2}$, $$\langle x,\lambda\rangle = x_{1}\lambda_{1}+x_{2}\lambda_{2},$$
$e(y):=e^{2\pi i y}$, and $a_{\lambda} \in \C$ some complex coefficients subject to
\begin{equation}
\label{eq:a(-lam)=conj(a(lam))}
a_{-\lambda}=\overline{a_{\lambda}},
\end{equation}
is a real-valued Laplace eigenfunction with eigenvalue
$E=E_{n}=4\pi^{2}n$, i.e. it satisfies the Helmholtz equation
\begin{equation}
\label{eq:Schrod eq}
\Delta f_{n}+Ef_{n}=0.
\end{equation}
Conversely, every real-valued function satisfying the equation \eqref{eq:Schrod eq}
is necessarily of the form \eqref{eq:fn toral Laplace eig} for some $n\in S$, and $\{a_{\lambda}\}_{\|\lambda\|^{2}=n}$ as above.

\vspace{2mm}

Given $n\in S$, the linear space of functions \eqref{eq:fn toral Laplace eig} subject to \eqref{eq:a(-lam)=conj(a(lam))}
is of real dimension $N_{n}$. The sequence $N_{n}$ is subject to large and erratic fluctuations.
However, its ``normal order" is
$$N_{n}=\log{n}^{\log{2}/2+o(1)},$$ though on average $N_{n} \sim \frac{1}{\kappa_{RL}}\cdot \sqrt{\log{n}}$ with $\kappa_{RL}>0$
the Ramanujan-Landau constant ~\cite{Landau}, and for every $\epsilon>0$
\begin{equation}
\label{eq:Nn small}
N_{n} = O(n^{\epsilon}),
\end{equation}
by an elementary argument.

We denote
\begin{equation*}
\Ec_{n} = \{ (\lambda_{1},\lambda_{2})\in\Z^{2}:\: \lambda_{1}^{2}+\lambda_{2}^{2}=n \}
\end{equation*}
to be the representations of $n$ as sum of two squares, or, what is equivalent,
$\Ec_{n}$ are all standard lattice points lying on the radius-$\sqrt{n}$ circle.
One may endow this space with a probability measure by assuming that the
$\{a_{\lambda}\}_{\lambda\in\Ec_{n}}$ are standard (complex) Gaussian\footnote{We work under
the convention that $a_{\lambda}=b_{\lambda}+ic_{\lambda}$, where the $b_{\lambda}$ and $c_{\lambda}$ are
standard real-valued Gaussians.} i.i.d. save to \eqref{eq:a(-lam)=conj(a(lam))},
turning $\{f_{n}\}_{n\in S}$ into a Gaussian {\em ensemble of random fields} ~\cite{ORW,RW2008}, all defined on $\Tb^{2}$,
usually referred to as ``Arithmetic Random Waves" ~\cite{KKW}. Alternatively, $f_{n}$ are unit variance stationary random fields on $\Tb^{2}$,
uniquely defined via their covariance function
\begin{equation}
\label{eq:rn covar func def}
r_{n}(x)=r_{n}(y,x+y):=\E[f_{n}(y)\cdot f_{n}(x+y)] = \frac{1}{N_{n}}\sum\limits_{\lambda\in\Ec_{n}} \cos(2 \pi \langle \lambda,x\rangle).
\end{equation}

\subsection{Defect}

The (total) defect of a smooth, not identically vanishing, function $g:\Tb^{2}\rightarrow\R$, (called ``signed area"
within the physics literature) is
\begin{equation*}
\Dpc(g):= \area(g^{-1}(0,+\infty))-\area(g^{-1}(-\infty,0)) = \int\limits_{\Tb^{2}}H(g(y))dy,
\end{equation*}
with $H(\cdot)$ denoting the sign function
\begin{equation}
\label{eq:H Heaviside}
H(y):=\begin{cases}
1 &y>0 \\ 0 &y=0 \\ -1 &y<0
\end{cases}.
\end{equation}
The defect of Laplace eigenfunctions was first addressed in the physics literature ~\cite{BGS} for random planar monochromatic
waves. A precise asymptotic expression for the defect variance, and a Central Limit Theorem was established,
along with generic nonlinear functionals, for the ensemble $\{T_{l}\}_{l\ge 1}$
of random Gaussian spherical harmonics ~\cite{MWvar,MWCLT}
with mathematical rigour. The $T_{l}:\Sc^{2}\rightarrow\R$ is the important ensemble of spherical random fields defined by
the covariance functions
$$\E[T_{l}(x)\cdot T_{l}(y)] = P_{l}(\cos (d(x,y))),$$ where $P_{l}(\cdot)$ are the Legendre polynomials and $d(\cdot, \cdot)$ is the spherical
distance; $T_{l}(\cdot)$ scales asymptotically like Berry's Random Waves around every point of $\Sc^{2}$, the main findings of
~\cite{MWvar,MWCLT} being consistent with ~\cite{BGS}, up to the said scaling.

\vspace{2mm}

We are interested in the defect of $f_{n}(\cdot)$ as in \eqref{eq:fn toral Laplace eig}.
We claim that for {\em every} such function $f_{n}$, the corresponding defect
\begin{equation}
\label{eq:tot def triv}
\Dpc(f_{n})\equiv 0
\end{equation}
vanishes, so the study of $\Dpc(f_{n})$ trivialises, and, accordingly, below we will pass to subdomains of $\Tb^{2}$.
First, if $n$ is odd, then for every $\lambda=(\lambda_{1},\lambda_{2})\in\Ec_{n}$,
necessarily precisely one of $\lambda_{1}$ and $\lambda_{2}$ is odd.
Hence, $f_{n}$ changes its sign under the involution $\tau:\Tb^{2}\rightarrow\Tb^{2}$ mapping $\cdot\mapsto \cdot+(1/2,1/2)$,
i.e. $$f_{n}(\tau x) = -f_{n}(x),$$ which readily implies $\Dpc(f_{n})= 0$. Otherwise, if $n$ is even, we may assume
w.l.o.g. that\footnote{Otherwise both the entries $\lambda_{1},\lambda_{2}$ are even, which yields that $f_{n}$ is invariant under
the involutions $\cdot\mapsto \cdot+(1/2,0)$ and $\cdot\mapsto\cdot+(0,1/2)$,
and we may pass from $n$ to $n/4$.} $n\equiv 2 (4)$, whence for all $\lambda\in\Ec_{n}$, both
$\lambda_{1},\lambda_{2}$ are odd, and then $f_{n}$ changes its sign under the involution $\rho:\Tb^{2}\rightarrow\Tb^{2}$
mapping $\cdot\mapsto \cdot+(1/2,0)$ (or $\cdot\mapsto \cdot+(0,1/2)$), also yielding $\Dpc(f_{n})= 0$.

It is therefore essential to pass to, possibly shrinking, {\em subdomains} of $\Tb^{2}$, most canonically, the radius-$s$ discs
$B_{x}(s)\subseteq\Tb^{2}$ centred at $x\in\Tb^{2}$, $0<s<1/2$, and $B(s):=B_{0}(s)$, with $s=s(n)$ allowed to depend on $n$,
(possibly $s=s(n)\rightarrow 0$). Since Quantum Chaos should exhibit itself
above Planck scale $s\gg \frac{1}{\sqrt{n}}$ ~\cite{Berry1977}, it
makes sense to take, as an example, $s=n^{-1/2+\epsilon}$, or, perhaps, replace
the $\epsilon$-power of $n$ with a slower growing function of $n$ (such as a power of $\log{n}$).
Our principal results concern the defect distribution corresponding to both the Arithmetic Random Waves
(random Gaussian toral eigenfunctions) in \S\ref{eq:Def var up low bnd ARW} below,
and individual deterministic cases, w.r.t. space average in \S\ref{sec:def var vanish space} below.

\subsection{Statement of principal results: defect variance for Arithmetic Random Waves}
\label{eq:Def var up low bnd ARW}

First, we take $f_{n}(\cdot)$ to be the Arithmetic Random Waves
(i.e. the random Gaussian model associated to \eqref{eq:fn toral Laplace eig}), and denote
\begin{equation}
\label{eq:Dns def}
\Dpc_{n;s}:= \frac{1}{\pi s^{2}}\int\limits_{B(s)}H(f_{n}(y))dy,
\end{equation}
where the normalisation makes $\Dpc_{n;s}$ invariant w.r.t. homotheties, and, by the stationarity of $f_{n}$,
the law of $\Dpc_{n;s}$ is independent of the centre of the disc (which is why we are may
assume that the disc on the r.h.s. of \eqref{eq:Dns def} is centred).
Since\hidefixme{I decided to remove all except that the expectation vanishes, as otherwise
the variance result is less than clear.}, for a given $y\in\Tb^{2}$, the law of $f_{n}(y)$ is symmetric around the origin, and $H(\cdot)$ is
odd, we have $\E[H(f_{n}(y))] \equiv 0$, and, by inverting the integral on the r.h.s. of
\eqref{eq:Dns def}, it is evident that for every $n\in S$ and $s>0$,
\begin{equation}
\label{eq:E[D]==0}
\E[\Dpc_{n;s}]=0.
\end{equation}
Our first principal result asserts that $\var(\Dpc_{n;s})\rightarrow 0$ as long as
the ball radius is above the Planck scale, i.e., $s\cdot
\sqrt{n}\rightarrow\infty$.

\begin{theorem}
\label{thm:uppper bnd det}

Fix $\epsilon >0 $ sufficiently small. For every $ 0< \delta < 4\epsilon $ one has
\begin{equation*}
\var\left(\Dpc_{n;s}\right) \ll \frac{1}{N_{n}^{\delta}}
\end{equation*}
uniformly for all $s>n^{-1/2+\epsilon}$. Equivalently,
\begin{equation*}
\sup\limits_{n\in S}\sup\limits_{s>n^{-1/2+\epsilon}} \var\left(\Dpc_{n;s}\right)\cdot N_{n}^{\delta}<+\infty.
\end{equation*}

\end{theorem}

If one is willing to excise a {\em thin}
sequence of energies, that is, a subsequence $S'$ of $S$ whose relative asymptotic density\footnote{A subset $S'\subseteq S$ is of relative
density $\kappa$ in $S$, if $$\lim\limits_{X\rightarrow\infty} \frac{\#S'(X)}{\#S(X)} = \kappa,$$ where for $\Ac\subseteq \Nb$ we define
$\Ac(X):=\{n\le X:\: n\in \Ac\}$.} in $S$ is $0$,
so that whatever generic energy levels are remaining satisfy certain arithmetic conditions explicated
in Theorem \ref{thm:var upper bnd expl} of \S\ref{sec:outline proofs} below,
then the asserted rate of decay is significantly more rapid, namely, faster than polynomial in $N_{n}$.

\begin{theorem}
\label{thm:ARW superpol up bnd}

For every $\epsilon>0$ there exists a subsequence $S'=S'(\epsilon)\subseteq S$ of energy levels of relative density $1$, so that,
along $n\in S'$, the inequality
\begin{equation}
\label{eq:var bound ARW gen}
\sup\limits_{s>n^{-1/2+\epsilon}} \var\left(\Dpc_{n;s}\right) \ll \frac{1}{N_{n}^{A}},
\end{equation}
holds for every $A>0$.

\end{theorem}

To the other end, we claim the following {\em lower} bound for $\var(\Dpc_{n;s})$ above Planck scale, valid for all $n\in S$.

\begin{theorem}
\label{thm:var def ARW low bnd}
Let $s=s(n)$ be a sequence of radii so that $T:=s\cdot \sqrt{n}\rightarrow\infty$.

\begin{enumerate}[a.]

\item For every $\delta>0$ there exists a sufficiently large number $A=A(\delta)$ so that
\begin{equation}
\label{eq:var>>1/N^A T^3+del}
\var\left(\Dpc_{n;s}\right) \gg \frac{1}{N_{n}^{A}\cdot T^{3+\delta}}.
\end{equation}

\item
\label{it:bnd away Bessel}
If, in addition, $2 \pi T$ is bounded away from the zeros of the Bessel $J_{1}$ function, then
\begin{equation}
\label{eq:var>>1/T^3}
\var\left(\Dpc_{n;s}\right) \gg \frac{1}{T^{3}}.
\end{equation}

\end{enumerate}

\end{theorem}

For comparison
of the generic upper bound \eqref{eq:var bound ARW gen} with the lower
bounds \eqref{eq:var>>1/N^A T^3+del} and \eqref{eq:var>>1/T^3}
(restricted to the regime $s>n^{-1/2+\epsilon}$ all the said bounds
hold) one should bear in mind \eqref{eq:Nn small}, i.e. that every
arbitrarily small positive power of $n$ dominates every power of
$N_{n}$.  It is well known that at infinity, the zeros of the Bessel
$J_{1}$ function are asymptotic to the arithmetic sequence
\begin{equation}
\label{eq:bessel zeros asympt}
\left\{\frac{\pi}{4}+\pi\cdot n\right\}_{n\ge 1}.
\end{equation}
The a fortiori meaning of the condition postulated by Theorem
\ref{thm:var def ARW low bnd}\ref{it:bnd away Bessel} is that $2 \pi T$ is bounded away by at least $\epsilon_{0}>0$ from the said sequence
\eqref{eq:bessel zeros asympt}, whence the conclusions apply (with constants depending on $\epsilon_{0}$).

\subsection{Statement of principal results: spatial defect distribution}

\label{sec:def var vanish space}

Rather than working with a Gaussian random field, we can take a sequence of {\em deterministic} eigenfunctions $f_{n}$ of the
form \eqref{eq:fn toral Laplace eig}, and study the defect distribution of $f_{n}$ restricted to $B_{x}(s)$, where $x$ is {\em random uniform} on $\Tb^{2}$, and $s$ is above Planck scale.
That is, given a function $f_{n}$ of the
form \eqref{eq:fn toral Laplace eig}, $x\in\Tb^{2}$ and $s>0$, we consider
\begin{equation}
\label{eq:Yfns def spac def}
Y_{f_{n},s}(x):= \frac{1}{\pi s^{2}}\int\limits_{B_{x}(s)}H(f_{n}(y))dy,
\end{equation}
the defect of $f_{n}$ restricted to $B_{x}(s)$.
Such an approach was recently taken by Sarnak ~\cite{Sarnak} and Humphries ~\cite{Humphries} for modular forms, and Granville-Wigman ~\cite{GW}
and Wigman-Yesha ~\cite{WiYe} for toral Laplace eigenfunctions \eqref{eq:fn toral Laplace eig},
in studying the {\em mass distribution} of the respective models, showing, in particular, that if there
exist discs observing unproportionately large or small $L^{2}$-mass of $f_{n}$, then these are not ``typical".

Of our principal interest here is the distribution
of the values of $Y_{f_{n},s}(\cdot)$ in \eqref{eq:Yfns def spac def} as $x$ distributes randomly uniformly on $\Tb^{2}$;
we denote accordingly the ``spatial defect expectation"
\begin{equation*}
\E_{\Tb^{2}}[Y_{f_{n},s}]:= \int\limits_{\Tb^{2}}Y_{f_{n},s}(x)dx,
\end{equation*}
and the ``spatial defect variance"
\begin{equation*}
\var_{\Tb^{2}}(Y_{f_{n},s}):= \int\limits_{\Tb^{2}} \left(Y_{f_{n},s}(x) - \E_{\Tb^{2}}[Y_{f_{n},s}]\right)^{2}dx.
\end{equation*}
The degeneracy argument identical to the argument we used to establish \eqref{eq:tot def triv}
that the total defect of every function \eqref{eq:fn toral Laplace eig} vanishes, yields that, in general,
the spatial defect expectation vanishes precisely, i.e., that
\begin{equation}
\label{eq:spatial exp vanish}
\E_{\Tb^{2}}[Y_{f_{n},s}]=0.
\end{equation}

\vspace{2mm}

In what follows, we will restrict ourselves to Bourgain's class
~\cite{Bo2013} of eigenfunctions
\begin{equation*}
\Bc_{n} = \left\{ f_{n}= \sum\limits_{\lambda\in\Ec_{n}}a_{\lambda}\cdot e(\langle x,\lambda\rangle):\: \forall\lambda\in \Ec_{n},
\,|a_{\lambda}|=1\, \text{ and } a_{-\lambda}=\overline{a_{\lambda}}  \right\}.
\end{equation*}
Our principal result concerning the spatial defect
distribution
asserts that for generic $n\in S$, and $f_{n}\in\Bc_{n}$ a Bourgain
class function, the spatial defect variance vanishes uniformly for $s$
slightly above Planck scale. Since $Y_{f_{n},s}$ is bounded, this is
equivalent to the statement that, in the said scenario, the proportion
of positive values of $f_{n}$ in ``most" discs of radius above Planck
scale is asymptotic to $1/2$ (see Lemma \ref{lem:var vanishing bnded
  rvs} below).  Despite that, what seems likely, the proof of the
principal result immediately below holds for a more general family of
flat eigenfunctions of the type considered in ~\cite{WiYe} (an event
of almost full Gaussian probability), we abandon the possible
generality for the sake of the elegance of presentation. That {\em
  some} flatness condition is {\em essential} for the defect variance
vanishing is asserted in Theorem \ref{thm:nonflat large def} to follow
immediately after the announced principal result.

\begin{theorem}
\label{thm:var upper bnd}
There exists a sequence $S''\subseteq S$ of relative density $1$, so that for all $\epsilon>0$ there exists $R=R(\epsilon)>0$
and $n_{0}=n_{0}(\epsilon)$ sufficiently large, so that for all $n>n_{0}$ with $n\in S''$,
\begin{equation*}
\var_{\Tb^{2}}(Y_{f_{n},s})  < \epsilon
\end{equation*}
holds uniformly for all $f_{n}\in\Bc_{n}$, $s>R/\sqrt{n}$. Equivalently\footnote{Formally, unrolling the definition
of the double limit below yields a slightly different, though equivalent to the above, statement, since it
is strongest for $R$ small.},
\begin{equation*}
\lim\limits_{\substack{R\rightarrow\infty\\n\rightarrow\infty,\,n\in S''}} \sup\limits_{\substack{s>R/\sqrt{n} \\ f_{n}\in\Bc_{n}}}\var_{\Tb^{2}}(Y_{f_{n},s}) = 0.
\end{equation*}
\end{theorem}

The arithmetic conditions on a sequence $S''$ as postulated in Theorem \ref{thm:var upper bnd}
will be explicated in \S\ref{sec:outline proof spatial}
below, as part of Theorem \ref{thm:derand expl}; they are more restrictive as compared to the subsequence $S'$ postulated in
Theorem \ref{thm:ARW superpol up bnd}. Finally, the result on the flatness being of essence for the spatial defect
variance vanishing announced above is stated, with radii vanishing {\em arbitrarily} slowly (or even fixed small radii).

\begin{theorem}
\label{thm:nonflat large def}
There exists a (thin) sequence $S'''\subseteq S$, a deterministic
sequence $\{f_{n}\}_{n\in S'''}$ of eigenfunctions
\eqref{eq:fn toral Laplace eig},
and numbers $\gamma, \epsilon_{0}>0$, so that the inequality
\begin{equation*}
\liminf\limits_{n\in S'''}\var_{\Tb^{2}}(Y_{f_{n},\Psi(n)}) > \epsilon_{0}
\end{equation*}
%holds for every function $\Psi:\Z_{>0}\rightarrow (0,1/2)$,
%subject to $\Psi(n) \le \gamma$ and $\Psi(n) n^{1/2} \to \infty$.
holds for every function $\Psi:\Z_{>0}\rightarrow (0, \min(\gamma,1/2))$,
subject to  $\Psi(n) n^{1/2} \to \infty$.
\end{theorem}

\subsection*{Acknowledgements}

We are indebted to Ze\'{e}v Rudnick for many stimulating discussions, and his comments
on an earlier version of this manuscript, in particular, pertaining to
Lemma \ref{lem:Diop approx roots} on Diophantine approximations.
The research leading to these results has received funding from the
European Research
Council under the European Union's Seventh Framework Programme
(FP7/2007-2013), ERC grant agreement n$^{\text{o}}$ 335141 (I.W. and
N.Y.).
P.K. was partially supported by
the
Swedish Research Council (2016-03701).

\section{Outline of the paper}

\subsection{Number Theoretic preliminaries}

Before we will be able to explain the essence of our arguments we will be required to bring forward
some arithmetic aspects of the lattice points $\Ec_{n}$.

\subsubsection{Angular equidistribution of lattice points}

First, we are interested in the {\em angular} distribution of $\Ec_{n}$. To this end we define the sequence
\begin{equation*}
\nu_{n}:=\frac{1}{N_{n}}\sum\limits_{\lambda\in\Ec_{n}}\delta_{\lambda/\sqrt{n}}
\end{equation*}
of probability measures on $\Sc^{1}\subseteq \R^{2}$, indexed by $n\in S$. It is well-known ~\cite{KK,EH,FKW}
that generically the angles of $\Ec_{n}$ are equidistributed, i.e. along a sequence $\{n\}\subseteq S$ of relative
density $1$,
\begin{equation}
\label{eq:nun=>arc len}
\nu_{n}\Rightarrow \frac{d\theta}{2\pi},
\end{equation}
where, as usual, $``\Rightarrow"$ stands for weak-$*$ convergence of probability measures, and
$\frac{d\theta}{2\pi}$ is the normalised arc-length measure on the unit circle. However, even under the (generic) assumption
$N_{n}\rightarrow\infty$, there exist sequences $\{n\}\subseteq S$ so that $\nu_{n}\Rightarrow\tau$ with $\tau$ different
than $\frac{d\theta}{2\pi}$; by definition, $\tau$ can be  any ``attainable"
probability measure on $\Sc^{1}$, e.g. the Cilleruelo measure ~\cite{Cil}
$$\tau=\frac{1}{4}\left(\delta_{\pm 1}+\delta_{\pm i}\right),$$ or
``intermediate'' measures
(e.g. measures supported on Cantor set, cf. ~\cite{KKW}); for a partial
classification see ~\cite{KW,Sartori}.

\begin{definition}
\label{def:equidist}
For a sequence $\{n\}\subseteq S$ we say that $\Ec_{n}$ are asymptotically equidistributed if \eqref{eq:nun=>arc len} holds.
\end{definition}

\subsubsection{Spectral correlations and quasi-correlations}
\label{sec:corr quasi-corr}

One of the key ingredients in ~\cite{KKW} was controlling the size of length-$6$ ``spectral correlations set". Given
$l\ge 3$, the length-$l$ spectral correlation set of the torus is the set
\begin{equation}
\label{eq:P_n_def}
\Pc_{n}(l):=\left\{(\lambda^{1},\ldots,\lambda^{l}) \in \Ec_{n}^{l}:\: \sum\limits_{j=1}^{l}\lambda^{j}=0 \right\}
\end{equation}
of $l$-tuples of lattice points in $\Ec_{n}$ summing up to $0$. Since, unless $n$ is divisible by $4$
(whence we can pass to $n/4$ in place of $n$), for $\lambda\in\Ec_{n}$,
the number of odd coordinates among $\lambda_{1},\lambda_{2}$ is $1$ or $2$ depending on the parity of $n$
(but independent of $\lambda\in\Ec_{n}$), for $l$ odd, the correlation sets
\begin{equation}
\label{eq:odd corr empty}
\Pc_{n}(l) = \varnothing
\end{equation}
are all empty ~\cite{BW} by a congruence
obstruction modulo $2$ argument, similar to the one yielding \eqref{eq:tot def triv}.
Otherwise, for $l$ even, the number of length-$l$ correlations
\begin{equation*}
\frac{1}{N_{n}^{l}}\cdot \#\Pc_{n}(l) = \int\limits_{\Tb^{2}}r_{n}(x)^{l}dx
\end{equation*}
is equal to the (normalized) moments of the covariance function \eqref{eq:rn covar func def} of the Arithmetic Random Waves.

Since for $l=2k$, all the ``diagonal" tuples $(\lambda^{1},-\lambda^{1},\ldots,\lambda^{k},-\lambda^{k})$ and their permutations
are in $\Pc_{n}(l)$, it implies the inequality
\begin{equation*}
\#\Pc_{n}(l) \gg N_{n}^{k}.
\end{equation*}
Conversely, Bombieri-Bourgain ~\cite{BB} proved, among other things, that, given $l=2k$ even, the inequality
\begin{equation}
\label{eq:corr<<N^k}
\#\Pc_{n}(l) \ll_{l} N_{n}^{k}
\end{equation}
holds for a generic sequence $\{n\}\subseteq S$; by invoking the usual diagonal argument, \eqref{eq:corr<<N^k} holds for
{\em all} $l$ even, along a generic sequence $\{n\}\subseteq S$.

\begin{definition}[Correlation-tame sequences of energies]
\label{def:B axiom corr}
We say sequence $S'\subseteq S$ is correlation-tame, if for every $l=2k\ge 6$ even, the inequality
\eqref{eq:corr<<N^k} holds true.
\end{definition}

In fact, Bombieri-Bourgain ~\cite{BB} proved a stronger property satisfied by the correlations of $\Ec_{n}$,
with $n$ generic, i.e. that a generic sequence in $S$ satisfies the following axiom $\Fc(\gamma)$ for some $0<\gamma<1/2$.

\begin{definition}[Axiom $\Fc(\gamma)$]
\label{def:axiom F(gamma)}

\begin{enumerate}

\item For $l\ge 4$, $n\in S$, we say that $(\lambda^{1},\ldots,\lambda^{l})\in\Ec_{n}^{l}$ is a minimal correlation,
if $\sum\limits_{j=1}^{l}\lambda^{j}=0$ and no proper subsum of $\sum\limits_{j=1}^{l}\lambda^{j}$ vanishes.

\item For $0<\gamma<1/2$ we say that a sequence $\{n\}\subseteq S$ satisfies the axiom $\Fc(\gamma)$, if for every
$l\ge 4$, the number of length-$l$ minimal correlations of $\Ec_{n}$ is at most $N_{n}^{\gamma\cdot l}$ for $n$ sufficiently
big.

\end{enumerate}

\end{definition}

As we will deal with moments of $r_{n}(\cdot)$ restricted to shrinking balls, we will find that, for our purposes,
the relevant notion is that of {\em quasi-correlations} ~\cite{BMW} (see \eqref{eq:part mom corr qcorr} below).
Given $n\in S$, $\epsilon>0$ and $l\ge 2$,
the length-$l$ quasi-correlation set is\footnote{Mind the slight abuse of notation as compared to ~\cite{BMW}}
\begin{equation*}
\Cc_{n}(l,\epsilon):=
\left\{(\lambda^{1},\ldots,\lambda^{l}) \in \Ec_{n}^{l}:\: 0<\left\| \sum\limits_{j=1}^{l}\lambda^{j}\right\|< n^{1/2-\epsilon} \right\};
\end{equation*}
note that, by the definition, $\Pc_{n}(l)$ and $\Cc_{n}(l,\epsilon)$ are disjoint.
It was shown ~\cite[Theorem 1.4]{BMW} that, given $l\ge 2$ and $\epsilon>0$,
the length-$l$ quasi-correlation set is empty $\Cc_{n}(l,\epsilon)=\varnothing$
along a generic sequence $\{n\}\subseteq S$, and, as it is the case of the correlation set, by a diagonal argument,
we may choose a density-$1$ subsequence $\{n\}\subseteq S$, so that along that sequence, for {\em every} $l\ge 2$,
$$\Cc_{n}(l,\epsilon)=\varnothing$$ holds true for $n$ sufficiently big (depending on $l$).

\vspace{2mm}

\begin{definition}[Axiom $\Ac(\epsilon)$ on sequences of energies]
\label{def:A axiom qcorr}
Given $\epsilon>0$ we say that a sequence $S'\subseteq S$ satisfies the axiom\footnote{Mind again an abuse of
notation compared to ~\cite{BMW}} $\Ac(\epsilon)$, if for every $l\ge 2$, the equality
$\Cc_{n}(l,\epsilon)=\varnothing$ holds for $n$ sufficiently big.
\end{definition}

\subsection{Outline of the proofs for Arithmetic Random Waves (theorems \ref{thm:uppper bnd det}-\ref{thm:var def ARW low bnd})}
\label{sec:outline proofs}

Here we assume that $\{f_{n}\}_{n\in S}$ are the (Gaussian) Arithmetic Random Waves.
Since it is possible to derive the identity
\begin{equation}
\label{eq:arcsineFormula}
\E[H(f_{n}(x))\cdot H(f_{n}(y))] = \frac{2}{\pi}\arcsin (r_{n}(x-y)),
\end{equation}
(cf. Lemma~\ref{lem:VarExact})
a straightforward manipulation with the definition \eqref{eq:Dns def} of $\Dpc_{n;s}$ and inverting the order of integration,
upon bearing in mind the stationarity of $f_{n}$, yields
the following {\em precise} expression for the defect variance:
\begin{equation}
\label{eq:var(defect) arcsin}
\var(\Dpc_{n;s}) = \frac{2}{\pi^{3} s^{4}}\int\limits_{B(s)\times B(s)}\arcsin(r_{n}(x-y))dxdy.
\end{equation}
Now we Taylor expand the arcsine around the origin (note that the
series converges absolutely at the endpoints $t=\pm 1$)
\begin{equation}
\label{eq:arcsin Taylor}
\arcsin(t) = \sum\limits_{k=0}^{\infty}a_{k}t^{2k+1},
\end{equation}
where all the (explicit) $a_{k}>0$ are {\em positive}, and substitute into \eqref{eq:var(defect) arcsin} to relate
between the defect variance and the moments of the covariance function restricted to $B(s)$:
\begin{equation}
\label{eq:var(defect) arcsin taylor}
\var(\Dpc_{n;s}) = \frac{2}{\pi^{3} s^{4}}\sum\limits_{k=1}^{\infty}a_{k}\cdot
\int\limits_{B(s)\times B(s)}r_{n}(x-y)^{2k+1}dxdy.
\end{equation}

\vspace{2mm}

We may in turn exploit the additive structure \eqref{eq:rn covar func
  def} to relate the said {\em odd} moments of $r_{n}(\cdot)$ to the
spectral correlations (and, implicitly, the quasi-correlations)
defined in \S\ref{sec:corr quasi-corr}:
\begin{equation}
\label{eq:part mom corr qcorr}
\int\limits_{B(s)\times B(s)}r_{n}(x-y)^{2k+1}dxdy=
\frac{ s^{2}}{N_{n}^{2k+1}}\sum\limits_{(\lambda^{1},\ldots,\lambda^{2k+1})\notin
\Pc_{n}(2k+1)} \frac{J_{1}(2 \pi s\cdot \| \lambda^{1}+\ldots + \lambda^{2k+1}\|)^{2}}
{ \| \lambda^{1}+\ldots + \lambda^{2k+1}\|^{2}},
\end{equation}
with $J_{1}(\cdot)$ the Bessel $J$ function of the first order,
so that to relate the defect variance to the spectral correlations and quasi-correlations (where, to obtain
\eqref{eq:part mom corr qcorr}, we separate the diagonal and use the observation \eqref{eq:odd corr empty}).
One may then substitute \eqref{eq:part mom corr qcorr} into \eqref{eq:var(defect) arcsin taylor} to obtain a more explicit
expression for $\var(\Dpc_{n;s})$, an absolutely convergent infinite series over all $(2k+1)$-tuples of lattice points.
If we assume further, that $s=n^{-1/2+\epsilon}$ (say), and a sequence $\{n\}\subseteq S$ satisfies the $\Ac(\delta)$
axiom with some $\delta<\epsilon$, then all the summands on the r.h.s. of \eqref{eq:part mom corr qcorr} are formally decaying
like a (small) power of $n$, faster than any power of $N_{n}$ (see \eqref{eq:Nn small}).

There is a subtlety with this outlined approach though, as controlling the decay rate in
this infinite series uniformly seems very difficult (if possible at all). Instead, we will only control finitely many
summands and bound the contribution of the higher moments. With this approach, we will encounter the odd moments
of the absolute value $|r_{n}(\cdot)|$ of the covariance rather than the moments of the covariance, that
we will reduce to a moment of higher order via Cauchy-Schwarz. Theorem \ref{thm:uppper bnd det} is the result of such an application
when capping the series at the first degree Taylor approximation of the arcsine \eqref{eq:arcsin Taylor}, whereas
Theorem \ref{thm:ARW superpol up bnd} caps it at an arbitrarily high degree Taylor approximation, depending on the required $A>0$ in
\eqref{eq:var bound ARW gen}, while also appealing to the correlation-tame property of a generic sequence of energies.
We will be able to prove the following result, which, since the claimed sequence $S'$ is generic, thanks to the results
mentioned in \S\ref{sec:corr quasi-corr}, clearly implies Theorem \ref{thm:ARW superpol up bnd}.

\begin{theorem}[Theorem \ref{thm:ARW superpol up bnd} with control over $S'(\epsilon)$]
\label{thm:var upper bnd expl}
Let $\epsilon>0$ be given,
and assume that $S'\subseteq S$ is a sequence of energy levels satisfying the axiom $\Ac(\delta)$
with some $\delta<\epsilon$, and is correlation-tame. Then the conclusions of Theorem \ref{thm:ARW superpol up bnd} hold, i.e.,
along $n\in S'$,
\begin{equation*}
\sup\limits_{s>n^{-1/2+\epsilon}} \var\left(\Dpc_{n;s}\right) \ll \frac{1}{N_{n}^{A}},
\end{equation*}
for every $A>0$.
\end{theorem}

For the lower bounds in Theorem \ref{thm:var def ARW low bnd} one also starts from \eqref{eq:var(defect) arcsin taylor} and
\eqref{eq:part mom corr qcorr}. Indeed, since the Taylor coefficients $a_{k}$ in \eqref{eq:var(defect) arcsin taylor} are all
positive, and, in hindsight, so are all the moments \eqref{eq:part mom corr qcorr} of $r_{n}(\cdot)$, it is sufficient to
bound any of these from below. If $T:=s\cdot \sqrt{n}$ happens to be bounded away from zeros of the Bessel $J_{1}$ function,
this readily yields the bound \eqref{eq:var>>1/T^3} of Theorem \ref{thm:var def ARW low bnd}. Most of our argument takes upon
the opposite situation when $T$ approaches one of the Bessel $J_{1}$ zeros, whence we need to 
rule out the, a priori unlikely, possibility of all the
terms $$2 \pi s\cdot \left\Vert \sum\limits_{j=1}^{2k+1} \lambda^{j}\right\Vert$$ conspiring around the Bessel zeros. 
To resolve this situation we exploit the higher order Taylor approximates, whence appealing to the deep W. Schmidt's {\em simultaneous Diophantine approximation} theorem ~\cite{Schmidt}, for example, approximating $\sqrt{5}$ by rational number for $k=1$ or $\sqrt{13}$ and $\sqrt{17}$ for $k=2$;
to attain $\frac{1}{T^{3+\delta}}$ as in \eqref{eq:var>>1/N^A T^3+del} we will need to focus on arbitrarily high $k$. 

Instead of using such a powerful result as in \cite{Schmidt}, one can try to significantly soften 
our techniques by bounding away from integers the values of the linear form $\Lc:\R^{K}\rightarrow\R$ given by $\Lc(x)=\sum\limits_{j=1}^{K}x_{j}\sqrt{p_{j}}$,
with a collection of distinct primes $p_{j}\equiv 1 (4)$ of our choice. An application of Khintchine's transference principle 
~\cite{Khintchine} (see also ~\cite[Theorem 5C on p. 99-100]{Schmidt Diophantine}) with Liouville's bound $$\left| \Lc(x)+b\right|\gg \frac{1}{\|x\|^{2^{K}-1}},$$ valid for all $x\in\Z^{K}\setminus\{0\}$, $b\in\Z$ ~\cite[Lemma 1A on p. 151]{Schmidt Diophantine}, yields information on the simultaneous approximation
of $\{\sqrt{p_{j}}\}$ by rational numbers. Unfortunately, the exponent, resulting from such an application, grows to infinity with $K$, 
which, to our best knowledge, undermines any attempt of the described type, and we thereby abandon it in favour of appealing to ~\cite{Schmidt}.

\subsection{Outline of the proofs for spatial fluctuations (Theorem \ref{thm:var upper bnd})}
\label{sec:outline proof spatial}

By a simple manipulation with the defect definition \eqref{eq:Yfns def spac def}
and integration order exchange it is straightforward to derive the expression
\begin{equation}
\label{eq:var spac ord int change}
%\var_{\Tb^{2}}(Y_{f_{n},s}) = \frac{1}{(\pi s^{2})^{2}} \int\limits_{\Tb^{2}\times \Tb^{2}} H(f_{n}(y)) H(f_{n}(z)) \cdot s^{2}W(\|y-z\|)dydz
%
\var_{\Tb^{2}}(Y_{f_{n},s}) = \frac{1}{(\pi s^{2})^{2}} \int\limits_{\Tb^{2}\times \Tb^{2}} H(f_{n}(y)) H(f_{n}(z)) \cdot s^{2}W(\|y-z\|/s)dydz
\end{equation}
for the spatial defect variance, where $W$ is a certain weight function (``circle-circle intersection function")
supported on $[0,2]$, and is $C^{1}$ on $(0,2)$. It is conceivable that the asymptotic vanishing of $\var_{\Tb^{2}}(Y_{f_{n},s})$
follows by a direct analysis of the r.h.s. of \eqref{eq:var spac ord int change}. However it seems very difficult, as the appearance of
$H(\cdot)$ on the r.h.s. of \eqref{eq:var spac ord int change} does not allow us to capitalise on the special additive structure
\eqref{eq:fn toral Laplace eig} of $f_{n}$, especially, in light of the discontinuity of $H(\cdot)$ at the origin (so, for example,
Taylor expanding $H(\cdot)$ around the origin is problematic).

\vspace{2mm}

We abandon such a direct approach, and instead notice that, since
the random variable $Y_{f_{n},s}$ is bounded (by $1$), the variance $\var_{\Tb^{2}}(Y_{f_{n},s})$ asymptotically vanishing is equivalent to
$Y_{f_{n},s}$ asymptotically vanishing with high probability (i.e. for ``most" of the ball centres on the torus), and recall
that, under certain flatness conditions on $f_{n}$ (certainly satisfied by all $f_{n}\in\Bc_{n}$) and arithmetic conditions on $n$
(in the spirit of the ones given in \S\ref{sec:corr quasi-corr} above), $f_{n}(\cdot)$ exhibits ~\cite{Bo2013,BW}
Gaussian spatial value distribution when averaged over the whole torus. Using these ``de-randomisation" techniques we will be able to prove
the result to follow immediately; unlike the results of ~\cite{Bo2013,BW} (and ~\cite{SartoriMass}), this is a second-order result (as opposed to a first order one). Moreover, since, unlike ~\cite{Bo2013,BW}, the Gaussian input for Theorem \ref{thm:derand expl} is not
inherently contained within its statement, it seems that a more direct approach might be possible
for proving Theorem \ref{thm:derand expl}.
Recall axiom $\Fc(\gamma)$ in Definition \ref{def:axiom F(gamma)}, and lattice points equidistribution in Definition \ref{def:equidist}.

\begin{theorem}[A variant of Theorem \ref{thm:var upper bnd} with control over $S''$]
\label{thm:derand expl}

Let $S''\subseteq S$ be a sequence of energy levels satisfying the axiom $\Fc(\gamma)$ for some $\gamma\in (0,1/2)$,
and assume further that the corresponding $\Ec_{n}$ are asymptotically equidistributed. Then the conclusions of Theorem \ref{thm:var upper bnd} apply along $S''$, i.e.
\begin{equation}
\label{eq:double lim var vanish}
\lim\limits_{\substack{R\rightarrow\infty\\n\rightarrow\infty,\,n\in S''}} \sup\limits_{\substack{s>R/\sqrt{n} \\ f_{n}\in\Bc_{n}}}\var_{\Tb^{2}}(Y_{f_{n},s}) = 0.
\end{equation}

\end{theorem}

Theorem \ref{thm:var upper bnd} is a direct consequence of Theorem \ref{thm:derand expl}, because
axiom $\Fc(\gamma)$ holds with some $\gamma\in (0,1/2)$
for ``generic" $n\in S$, and $\Ec_{n}$ is asymptotically distributed for ``generic" $n\in S$ in the sense of Definition \ref{def:axiom F(gamma)}.
The proof of Theorem \ref{thm:derand expl} proceeds in three steps. First, we reduce proving \eqref{eq:double lim var vanish} uniformly
for $s>R/\sqrt{n}$ to proving for $s=R/\sqrt{n}$ only, via an analogue
of the Geometric-Integral Sandwich,
first introduced
in ~\cite{SodinSPB,NSNodal}, adapted to our settings. Next, we exploit the said spatial Gaussianity of $f_{n}(\cdot)$
in order to reduce the variance vanishing to the analogous result for the limit random field, which, by the equidistribution
assumption for $\Ec_{n}$ of Theorem \ref{thm:derand expl}, is the Gaussian random field of planar isotropic monochromatic waves
(it is ``Berry's Random Wave Model", uniquely defined by its covariance function $J_{0}(\|x\|)$).

It then remains to evaluate
the variance of the defect for the limit Gaussian random field restricted to a compact domain (e.g. the unit square), which, in spirit,
is already contained in ~\cite{MWvar} (and predicted by ~\cite{BGS}), where a rapid decay rate is asserted. This result is the only
use of the equidistribution assumption, and it should be not too
technically demanding to remove this assumption,
as long as some non-degeneracy for the limit Gaussian field is imposed (e.g. it cannot include the most degenerate ``Cilleruelo"
case), though it benefits us in no way if we are only interested in a density-$1$ sequences of energy levels.
Our main result \eqref{eq:double lim var vanish} is ineffective in terms of rate of decay for $\var_{\Tb^{2}}(Y_{f_{n},s})$,
as the convergence of the spatial distribution
of $f_{n}$ to the Gaussian is ineffective.

\subsection{Outline of constructing functions with non-vanishing defect variance (Theorem \ref{thm:nonflat large def})}

The prevailing symmetry obstruction, dictating that for the standard torus, the total defect of any Laplace eigenfunction
vanishes precisely does not persist for the non-standard tori. We exploit the hexagonal torus, so that to construct a single
Laplace eigenfunction with total defect non-vanishing, and scale it to obtain a sequence of eigenfunctions of arbitrarily high
energy, with defect growing on large fragments of the torus, above the Planck scale. We then mimic that situation
on the standard torus, by appealing to the Pell equation $x^{2}-3y^{2}=1$, yielding solutions approximating the
hexagonal toral eigenfunctions on the standard torus.

\subsection{Outline of the paper}

Section \ref{sec:ARW} is dedicated to giving the proofs for all the results concerning the defect of the Arithmetic Random Waves
(theorems \ref{thm:uppper bnd det}-\ref{thm:var def ARW low bnd}), appealing among the rest to Diophantine approximations.
In section \ref{sec:Bourgain de-randomization} Bourgain's de-randomization method will be invoked to prove Theorem \ref{thm:var upper bnd} dealing with the spatial defect variance vanishing for the flat functions.
Finally, a sequence of ``esoteric"
non-flat functions with spatial defect variance non-vanishing will be constructed in section \ref{sec:def non vanish},
by first constructing eigenfunctions with the analogous properties defined on the {\em hexagonal} torus (as opposed
to the standard torus).

\section{The defect of Arithmetic Random Waves: proof of theorems \ref{thm:uppper bnd det}-\ref{thm:var def ARW low bnd}}
\label{sec:ARW}

\subsection{Preliminary lemmas}

Let $f_{n}(\cdot)$ be the Arithmetic Random Wave corresponding to \eqref{eq:fn toral Laplace eig}, so that $f_{n}(\cdot)$ is a unit variance stationary Gaussian random field with covariance function
\eqref{eq:rn covar func def}. We first establish the precise expression \eqref{eq:var(defect) arcsin} for the variance of the defect $\Dpc_{n;s}$.

\begin{lemma}
	\label{lem:VarExact}
	We have
	\begin{equation*}
	\mathrm{Var}\left(\Dpc_{n;s}\right)=\frac{2}{\pi^{3}s^{4}}\int\limits _{B\left(s\right)\times B\left(s\right)}\arcsin\left(r_{n}\left(x-y\right)\right)dxdy.\label{eq:VarExactFormula}
	\end{equation*}
\end{lemma}

\begin{proof}
It is a well-known fact (see, e.g., \cite{Rice1,Rice2}) that every bivariate centred Gaussian random
vector $\left(X,Y\right)$ with covariance matrix
\[
\Sigma=\begin{pmatrix}1 & r\left(x,y\right)\\
r\left(x,y\right) & 1
\end{pmatrix}
\]
satisfies
\[
\mathbb{E}\left[H(X)\cdot H(Y)\right]=\frac{2}{\pi}\arcsin\left(r\left(x,y\right)\right).
\]
Hence, the identity \eqref{eq:arcsineFormula} follows letting $ X=f_n(x)$, $Y=f_n(y) $.
	
By the vanishing of the defect expectation \eqref{eq:E[D]==0}, we have
\begin{equation}
\label{eq:var_expansion}
\text{Var}\left(\Dpc_{n;s}\right) =\mathbb{E}\left[\Dpc_{n;s}^{2}\right]=\frac{1}{\left(\pi s^{2}\right)^{2}}\mathbb{E}\left[\int_{B\left(s\right)\times B\left(s\right)}H\left(f_{n}\left(x\right)\right)\cdot H\left(f_{n}\left(y\right)\right)dxdy\right].
\end{equation}
Changing the order of expectation and integration in \eqref{eq:var_expansion} together with the identity \eqref{eq:arcsineFormula} gives the desired formula for the defect variance.
\end{proof}
As we will see below, the defect variance $\var\left(\Dpc_{n;s}\right)$ is intimately related to the (restricted) moments of the covariance function $r_{n}(\cdot)$. The following
lemma gives a useful arithmetic formula for these moments.

\begin{lemma}
	\label{lem:CovFormulaLemma}Let $l\ge1$. We have
	\begin{equation}
\label{eq:CovFormula}
	\int\limits _{B\left(s\right)\times B\left(s\right)}r_{n}\left(x-y\right)^{l}dxdy=\left(\pi s^{2}\right)^{2}\frac{\#\Pc_{n}\left(l\right)}{N_{n}^{l}}+\frac{s^{2}}{N_{n}^{l}}\sum_{\left(\lambda^{1},\dots,\lambda^{l}\right)\notin\Pc_{n}\left(l\right)}\frac{J_{1}\left(2\pi s\left\Vert\lambda^{1}+\dots+\lambda^{l}\right\Vert\right)^{2}}{\left\Vert\lambda^{1}+\dots+\lambda^{l}\right\Vert^{2}}.
	\end{equation}
Moreover, if $l=2k+1$, then by \eqref{eq:odd corr empty} we have $\#\Pc_{n}(l)=0$, so that
\eqref{eq:CovFormula} reads \eqref{eq:part mom corr qcorr}.
\end{lemma}

\begin{proof}
	Expanding the covariance function \eqref{eq:rn covar func def}, and recalling the definition \eqref{eq:P_n_def} of $ \Pc_{n}\left(l\right) $, we obtain
	\begin{align*}
	\int\limits _{B\left(s\right)\times B\left(s\right)}r_{n}\left(x-y\right)^{l}dxdy & =\frac{1}{N_{n}^{l}}\int\limits _{B\left(s\right)\times B\left(s\right)}\sum_{\lambda^{1},\dots,\lambda^{l}\in\Ec_{n}}e\left(\left\langle \lambda^{1}+\dots+\lambda^{l},x-y\right\rangle \right)dxdy\\
	& =\left(\pi s^{2}\right)^{2}\frac{\#\Pc_{n}\left(l\right)}{N_{n}^{l}} +\frac{1}{N_{n}^{l}}\sum_{\left(\lambda^{1},\dots,\lambda^{l}\right)\notin\Pc_{n}\left(l\right)}\left|\int_{B\left(s\right)}e\left(\left\langle \lambda^{1}+\dots+\lambda^{l},x\right\rangle \right)dx\right|^{2}.
	\end{align*}
	Formula (\ref{eq:CovFormula}) now follows from the identity
	\[
	\int\limits _{B\left(s\right)}e\left(\left\langle v,x\right\rangle \right)dx=\frac{sJ_{1}\left(2\pi s\left\Vert v \right\Vert\right)}{\left\Vert v \right\Vert}.
	\]
\end{proof}
\subsection{Upper bounds}

We now turn to prove the upper bounds for $\text{Var}\left(\Dpc_{n;s}\right)$.
We begin with the proof of  Theorem \ref{thm:uppper bnd det}.
% \hrule
% new version below
% \hrule
\begin{proof}[Proof of Theorem \ref{thm:uppper bnd det}]
By Lemma \ref{lem:VarExact} and the elementary bound
$
\arcsin x
=
x + O(x^{2})
$
we have
\begin{equation}
  \label{eq:cubic-arcsin-bound}
\text{Var}\left(\Dpc_{n;s}\right)
=
\frac{2}{\pi^{3}s^{4}}
\left(
\int\limits_{B\left(s\right)\times  B\left(s\right)}
 r_{n}\left(x-y\right) dxdy
+
O \left(
\int\limits_{B\left(s\right)\times  B\left(s\right)}
 r_{n}\left(x-y\right)^{2} dxdy
\right)
\right)
\end{equation}
By Lemma \ref{lem:CovFormulaLemma},
$$
\int\limits_{B\left(s\right)\times  B\left(s\right)}
 r_{n}\left(x-y\right) dxdy
=
\frac{s^{2}}{N_{n}}
\sum_{\lambda \in \Ec_{n}}
\frac{J_{1}\left(2\pi s\left\Vert\lambda
    \right\Vert\right)^{2}}{\left\Vert\lambda \right\Vert^{2}}
$$
which, using the bound,

\begin{equation}
  \label{eq:bessel-one-bound}
J_{1}\left(x\right)
\ll\min\left\{  x^{-1/2},x\right\}
\end{equation}
(see formulas (9.1.7) and (9.2.1) in \cite{AS}) is
$$
\ll
\frac{s^{2}}{N_{n}} \frac{N_{n}}{s \left\Vert\lambda \right\Vert^{3}}
=
\frac{s^{4}}{ (s n^{1/2})^{3}} \le s^{4} n^{-3\epsilon}
$$ for all $ s>n^{-1/2+\epsilon} $.
We find that the contribution from the first  integral
on the r.h.s. of
\eqref{eq:cubic-arcsin-bound} is $\ll n^{-3 \epsilon}$.
	
We next evaluate the second integral on the r.h.s. of
\eqref{eq:cubic-arcsin-bound}. By Lemma \ref{lem:CovFormulaLemma},
	we have
	\begin{align}
	\int\limits _{B\left(s\right)\times B\left(s\right)}r_{n}\left(x-y\right)^{2}dxdy & =\frac{\pi^{2}s^{4}}{N_{n}}+\frac{s^{2}}{N_{n}^{2}}\sum_{\lambda^{1}\ne\lambda^{2}\in \Ec_{n}}\frac{J_{1}\left(2\pi s\left\Vert\lambda^{1}-\lambda^{2}\right\Vert\right)^{2}}{\left\Vert\lambda^{1}-\lambda^{2}\right\Vert^{2}},\label{eq:covSecondMoment}
	\end{align}
	where we used the fact that $\left(\lambda^{1},\lambda^{2}\right)\in\Pc_{n}\left(2\right)$
	if and only if $\lambda^{1}=-\lambda^{2},$ and in particular
	\[
	\#\Pc_{n}\left(2\right)=N_{n},
	\] and the symmetry $ \lambda \in \Ec_{n} \iff - \lambda \in
        \Ec_{n}. $
Again using the bound \eqref{eq:bessel-one-bound}
% 	By the bound \[J_{1}\left(x\right)\ll\min\left\{
%           x^{-1/2},x\right\}\](see formulas (9.1.7) and (9.2.1) in
%         \cite{AS}),
	we have
	\[
	\sum_{\lambda^{1}\ne\lambda^{2}\in\Ec_{n}}\frac{J_{1}\left(2\pi s\left\Vert\lambda^{1}-\lambda^{2}\right\Vert\right)^{2}}{\left\Vert\lambda^{1}-\lambda^{2}\right\Vert^{2}}\ll\sum_{\lambda^{1}\ne\lambda^{2}\in\Ec_{n}}\min\left\{ \frac{1}{s\cdot \left\Vert\lambda^{1}-\lambda^{2}\right\Vert^{3}},s^{2}\right\} ,
	\]
	and therefore, for any $ 0<\eta<1/2 $, we have
	\begin{equation}
	\sum_{\lambda^{1}\ne\lambda^{2}\in\Ec_{n}}\frac{J_{1}\left(2 \pi s\left\Vert\lambda^{1}-\lambda^{2}\right\Vert\right)^{2}}{\left\Vert\lambda^{1}-\lambda^{2}\right\Vert^{2}}\ll s^{2}\sum_{\substack{\lambda^1,\lambda^2\in\Ec_{n} \\ 0<\left\Vert\lambda^{1}-\lambda^{2}\right\Vert <n^{1/2-\eta}}}1+s^{-1}\sum_{\substack{\lambda^1,\lambda^2 \in \Ec_{n} \\ \left\Vert\lambda^{1}-\lambda^{2}\right\Vert\ge n^{1/2-\eta}}}\frac{1}{\left\Vert\lambda^{1}-\lambda^{2}\right\Vert^{3}}.\label{eq:LatticeSum}
	\end{equation}
	
We estimate the sums on the r.h.s. of \eqref{eq:LatticeSum} separately. The second sum on the r.h.s. of \eqref{eq:LatticeSum} can be bounded trivially:
\begin{equation}
\label{eq:trivial_bound}
s^{-1}\sum_{\substack{\lambda^1,\lambda^2 \in\Ec_{n} \\ \left\Vert\lambda^{1}-\lambda^{2}\right\Vert\ge n^{1/2-\eta}}}\frac{1}{\left\Vert\lambda^{1}-\lambda^{2}\right\Vert ^{3}}\le N_{n}^{2}s^{-1}\left(n^{1/2-\eta}\right)^{-3},
\end{equation}
whereas the first sum on the r.h.s. of \eqref{eq:trivial_bound} is the number of ``close-by pairs", bounded in \cite{GW} (see Theorem 1.8 there and the remark following it) by
\begin{equation}
\label{eq:GW}
\sum_{\substack{\lambda^1,\lambda^2\in\Ec_{n} \\ 0<\left\Vert\lambda^{1}-\lambda^{2}\right\Vert<n^{1/2-\eta }}}1 \ll N_{n}^{2-\tau \eta}
\end{equation} for any $ \tau < 4 $ and $ \eta>0 $ sufficiently small.

Substituting the bounds \eqref{eq:trivial_bound} and \eqref{eq:GW}
into \eqref{eq:LatticeSum}, and then back into
\eqref{eq:covSecondMoment}, we obtain the bound
\begin{equation}
\int\limits _{B\left(s\right)\times
  B\left(s\right)}r_{n}\left(x-y\right)^{2}dxdy\ll
s^{4}N_{n}^{-1}+sn^{-3/2+3\eta}+s^{4}N_{n}^{-\tau \eta}
=
s^{4}( N_{n}^{-1}+n^{3 \eta}/(sn^{1/2})^{3}+N_{n}^{-\tau \eta}).
\label{eq:secondMomentBound}
\end{equation}

Let $ 0<\delta < 4\epsilon $, and write $ \delta = \tau \eta $ where $ \tau < 4 $ and $ \eta < \epsilon $. Then \eqref{eq:secondMomentBound}, together with (\ref{eq:cubic-arcsin-bound}), the bound
\eqref{eq:Nn small}, and the previous bound on the first integral on the
r.h.s. of
(\ref{eq:cubic-arcsin-bound}), gives
$
\text{Var}\left(\Dpc_{n;s}\right)\ll N_{n}^{-\delta}
$
uniformly for all $s>n^{-1/2+\epsilon}$, completing the proof of Theorem \ref{thm:uppper bnd det}.

\end{proof}

We now prove Theorem \ref{thm:var upper bnd expl} which, as argued above, immediately implies Theorem \ref{thm:ARW superpol up bnd}.

\begin{proof}[Proof of Theorem \ref{thm:var upper bnd expl}]
	Recall that the Taylor series of $\arcsin\left(t\right)$ is given
	by \eqref{eq:arcsin Taylor}
	where
	\begin{equation}
	\label{eq:a_k}
	a_{k}=\frac{1}{2^{2k}} \binom{2k}{k} \frac{1}{2k+1},
	\end{equation}
	so that by Stirling's approximation $a_{k}\sim\frac{1}{2\sqrt{\pi}}k^{-3/2}$,
	and the convergence is uniform on $\left[-1,1\right].$ In particular
	for $K\ge0$, the Taylor polynomial of $\arcsin\left(t\right)$ is
	given by
	\begin{equation}
	\arcsin\left(t\right)=\sum_{k=0}^{K}a_{k}t^{2k+1}+O\left(\left|t\right|^{2K+3}\right).\label{eq:ArcsineTaylorPol}
	\end{equation}
	Substituting \eqref{eq:ArcsineTaylorPol} into \eqref{eq:var(defect) arcsin}
	yields
	\begin{equation}
	\text{Var}\left(\Dpc_{n;s}\right)  =\frac{2}{\pi^{3}s^{4}}\sum_{k=0}^{K}a_{k}\int\limits _{B\left(s\right)\times B\left(s\right)}r_{n}\left(x-y\right)^{2k+1}dxdy\label{eq:VarExpansion}
	 +O\left(\frac{1}{s^{4}}\int\limits _{B\left(s\right)\times B\left(s\right)}\left|r_{n}\left(x-y\right)\right|^{2K+3}dxdy\right).
	\end{equation}

Let $ 0\le k\le K $. Recall the identity \eqref{eq:part mom corr qcorr},
and that $ n\in S' $ where
the sequence $S' \subseteq S$ satisfies the axiom ${\Ac}\left(\delta\right)$ as in Definition \ref{def:A axiom qcorr}, so that the condition $\left(\lambda^{1},\dots,\lambda^{2k+1}\right)\notin\Pc_{n}\left(2k+1\right)$ in \eqref{eq:part mom corr qcorr}
implies that
\begin{equation}
\left\Vert\lambda^{1}+\dots+\lambda^{2k+1}\right\Vert\gg_{K}n^{1/2-\delta}.\label{eq:CorrelationLowerBd}
\end{equation}
	Substituting the bound \eqref{eq:CorrelationLowerBd} together with the bound \eqref{eq:bessel-one-bound}
%	\begin{equation}
%	J_{1}\left(x\right)\ll x^{-1/2}\label{eq:BesselDecay}
%	\end{equation}
	into \eqref{eq:part mom corr qcorr}, we get that
	\begin{align}
	\frac{1}{s^{4}}\int\limits _{B\left(s\right)\times B\left(s\right)}r_{n}\left(x-y\right)^{2k+1} \ll & \frac{1}{s^{3}N_{n}^{2k+1}}\sum_{\left(\lambda^{1},\dots,\lambda^{2k+1}\right)\notin\Pc_{n}\left(2k+1\right)}\frac{1}{\left\Vert\lambda^{1}+\dots+\lambda^{2k+1}\right\Vert^{3}}
	\ll_{K} s^{-3} n^{-3/2+3\delta}\nonumber\\
	\le& n^{-3\left(\epsilon-\delta\right)}\label{eq:odd_moment_bound}
	\end{align}
	uniformly for $  s>n^{-1/2+\epsilon}$. We can now use \eqref{eq:odd_moment_bound} to bound the summation in the variance formula (\ref{eq:VarExpansion}), which gives
	\begin{equation}
	\text{Var}\left(\Dpc_{n;s}\right)\ll_{K}n^{-3\left(\epsilon-\delta\right)}+ \frac{1}{s^{4}}\int\limits _{B\left(s\right)\times B\left(s\right)}\left|r_{n}\left(x-y\right)\right|^{2K+3}dxdy.\label{eq:MainMoments}
	\end{equation}
	
To control the $\left(2K+3\right)$'th moment of the absolute value
of $r_{n}(\cdot)$, we use the Cauchy-Schwarz inequality to discard the
absolute value:
\begin{equation}
\int\limits _{B\left(s\right)\times B\left(s\right)}\left|r_{n}\left(x-y\right)\right|^{2K+3}dxdy\le\pi s^{2}\left(\int_{B\left(s\right)\times B\left(s\right)}r_{n}\left(x-y\right)^{4K+6}dxdy\right)^{1/2}.\label{eq:OddMomentsIneq}
	\end{equation}
	By  Lemma \ref{lem:CovFormulaLemma}, we have
	\begin{align}
	\label{eq:4k+6_moment}
	\frac{1}{s^{4}}\int\limits _{B\left(s\right)\times B\left(s\right)}r_{n}\left(x-y\right)^{4K+6}dxdy=&\pi^{2}\frac{\#\Pc_{n}\left(4K+6\right)}{N_{n}^{4K+6}}\nonumber\\
	+&\frac{1}{s^{2}N_{n}^{4K+6}}\sum_{\left(\lambda^{1},\dots,\lambda^{4K+6}\right)\notin\Pc_{n}\left(4K+6\right)}\frac{J_{1}\left(2\pi s\left\Vert\lambda^{1}+\dots+\lambda^{4K+6}\right\Vert\right)^{2}}{\left\Vert\lambda^{1}+\dots+\lambda^{4K+6}\right\Vert^{2}}.
	\end{align}
	Since $ S' $ is correlation-tame (Definition \ref{def:B axiom corr}), we have $\#\Pc_{n}\left(4K+6\right)\ll_{K} N_{n}^{2K+3}$. This, together with
    \eqref{eq:4k+6_moment} and the estimate \eqref{eq:bessel-one-bound}, yields
	\begin{equation}
	\label{eq:4K+6}
	\frac{1}{s^{4}}\int\limits _{B\left(s\right)\times B\left(s\right)}r_{n}\left(x-y\right)^{4K+6}dxdy \ll_{K}\frac{1}{N_{n}^{2K+3}}+\frac{1}{s^{3}N_{n}^{4K+6}}\sum_{\left(\lambda^{1},\dots,\lambda^{4K+6}\right)\notin\Pc_{n}\left(4K+6\right)}\frac{1}{\left\Vert\lambda^{1}+\dots+\lambda^{4K+6}\right\Vert^{3}}.
	\end{equation}
	By the lower bound \eqref{eq:CorrelationLowerBd}, we have
	\begin{equation}
\frac{1}{s^{3}N_{n}^{4K+6}}\sum_{\left(\lambda^{1},\dots,\lambda^{4K+6}\right)\notin\Pc_{n}\left(4K+6\right)}\frac{1}{\left\Vert\lambda^{1}+\dots+\lambda^{4K+6}\right\Vert^{3}}\ll_{K} s^{-3} n^{-3/2+3\delta}\le n^{-3\left(\epsilon-\delta\right)}\label{eq:SmallTerms}
	\end{equation}
	uniformly for $ s>n^{-1/2+\epsilon} $. Substituting the bound \eqref{eq:SmallTerms} into \eqref{eq:4K+6} and bearing in mind \eqref{eq:Nn small}
	gives
	\begin{equation}
	\frac{1}{s^{4}}\int\limits _{B\left(s\right)\times B\left(s\right)}r_{n}\left(x-y\right)^{4K+6}dxdy\ll_{K}\frac{1}{N_{n}^{2K+3}}.\label{eq:BigTerm}
	\end{equation}
	
	Finally, we substitute the bound (\ref{eq:BigTerm}) into  (\ref{eq:OddMomentsIneq}), and then into \eqref{eq:MainMoments}. Using again  \eqref{eq:Nn small}, we get that \[ \text{Var}\left(\Dpc_{n;s}\right) \ll_{K}\frac{1}{N_{n}^{K+3/2}}.\] This completes the proof of Theorem \ref{thm:var upper bnd expl}, since $ K $ can be taken arbitrarily large.
\end{proof}

\subsection{Lower bound}

In order to prove the lower bound for $\text{Var}\left(\Dpc_{n;s}\right)$
stated in Theorem \ref{thm:var def ARW low bnd}, we will require a result
on Diophantine approximation by multiples of square roots of prime numbers.
 For $t\in \R$,
we denote $\langle t\rangle $ to be the distance of $t$ to the nearest integer number, and let
\begin{equation}
\label{eq:PK primes def}
P_K := \left\{ p\,\,\text{prime}:\,p\equiv1\,\left(\text{mod}\,4\right),\,p\le K\right\}
\end{equation}
denote the set of primes $ p\le K $ congruent to $ 1 $ modulo $ 4 $.

\begin{lemma}
\label{lem:Diop approx roots}
Let $K>1$ be an integer, and let $\epsilon>0$.
For every integer $q\ge1,$ we have
\begin{equation}
\label{eq:q sqrt(p)}
	\max\limits_{p\in P_K}\left< q\sqrt{p}\right> \gg_{K,\epsilon}q^{-\frac{2\log K}{K}-\epsilon}.
\end{equation}

\end{lemma}

The proof of Lemma \ref{lem:Diop approx roots} will invoke two classical results from the theory of Diophantine approximation:
Besicovich's theorem on the linear independence over $\Q$ of the square roots of distinct square-free positive integers, and Schmidt's theorem on simultaneous Diophantine approximation, that, for the reader's convenience, we cite next, in the form used subsequently.

\begin{theorem}[Besicovitch \cite{Bes}]
	\label{thm:Besicovitch}Let $q_{1},\dots,q_{m}$ be distinct squarefree
	positive integers. The numbers $\sqrt{q_{1}},\dots,\sqrt{q_{m}}$
	are linearly independent over $\Q$.
\end{theorem}

\begin{theorem}[Schmidt \cite{Schmidt}]
	\label{thm:Schmidt}Let $\alpha_{1},\dots,\alpha_{m}$ be real algebraic
	numbers so that $1,\alpha_{1},\dots,\alpha_{m}$ are linearly
	independent over the rationals. Then for every $\epsilon>0$ and for
	every integer $q\ge1$, we have
	\[
	\max_{1\le i\le m}\left < q\alpha_{i}\right > \gg q^{-1/m-\epsilon}
	\]
	where the implied constant depends on $\epsilon$ and on $\alpha_{1},\dots,\alpha_{m}$.
\end{theorem}

\begin{proof}[Proof of Lemma \ref{lem:Diop approx roots}]
	By Theorem \ref{thm:Besicovitch}, the elements of the set
	$\left\{ 1\right\} \cup \left\{\sqrt{p}: p \in P_K \right\}$
	are linearly independent over the rationals. Since $\#P_k \sim \frac{K}{2\log K}$ as $ K\to\infty $, the bound \eqref{eq:q sqrt(p)} follows from Theorem \ref{thm:Schmidt}.
\end{proof}

We are finally in a position to prove Theorem \ref{thm:var def ARW low bnd}.

\begin{proof}[Proof of Theorem \ref{thm:var def ARW low bnd}]
	Recall that substituting the Taylor series of the arcsine function (\ref{eq:arcsin Taylor})
	in (\ref{eq:var(defect) arcsin}) gives formula \eqref{eq:var(defect) arcsin taylor}:
	\[
	\text{Var}\left(\Dpc_{n;s}\right)=\frac{2}{\pi^{3}s^{4}}\sum_{k=0}^{\infty}a_{k}\int\limits _{B\left(s\right)\times B\left(s\right)}r_{n}\left(x-y\right)^{2k+1}dxdy,
	\]
	where $ a_k $ are given by \eqref{eq:a_k}, and in particular $a_{k}>0$, $a_{0}=1$, and $a_{k}\sim\frac{1}{2\sqrt{\pi}}k^{-3/2}$.
	Hence, Lemma \ref{lem:CovFormulaLemma} yields
	\begin{equation} \text{Var}\left(\Dpc_{n;s}\right)=\frac{2}{\pi^{3}s^{2}}\sum_{k=0}^{\infty}\frac{a_{k}}{N_{n}^{2k+1}}\sum_{\left(\lambda^{1},\dots,\lambda^{2k+1}\right)\notin\Pc_{n}\left(2k+1\right)}\frac{J_{1}\left(2\pi s\left\Vert\lambda^{1}+\dots+\lambda^{2k+1}\right\Vert\right)^{2}}{\left\Vert\lambda^{1}+\dots+\lambda^{2k+1}\right\Vert^{2}}.\label{eq:VarFormula}
	\end{equation}
	By the positivity of the coefficients $a_{k},$ we may obtain a lower bound by discarding all
terms in \eqref{eq:VarFormula} but one with $k = 0$:
	\begin{equation}
	\text{Var}\left(\Dpc_{n;s}\right)\ge\frac{2}{\pi^{3}}\frac{J_{1}\left(2\pi T\right)^{2}}{T^{2}}.\label{eq:VargeJ1}
	\end{equation}
	
	Recall that for large $z$, we have \cite[formula (9.2.1)]{AS}
	\begin{equation}
	J_{1}\left(z\right)=\sqrt{\frac{2}{\pi z}}\cos\left(z-\frac{3}{4}\pi\right)+O\left(\frac{1}{z^{3/2}}\right),\label{eq:BesselAysmp}
	\end{equation}
	so that
	\begin{equation}
	J_{1}\left(2\pi T\right)=\pi^{-1}T^{-1/2}\cos\left(\left(2T-\frac{3}{4}\right)\pi\right)+O\left(T^{-3/2}\right).\label{eq:J_1(T)}
	\end{equation}
	We write
	\begin{equation}
	2T-\frac{1}{4}=t+\rho\label{eq:TApprox}
	\end{equation}
	where $t=t\left(T\right)\in\mathbb{Z}$ and $\left|\rho\right|\le1/2$,
	so that
	\begin{equation}
	\left|\cos\left(\left(2T-\frac{3}{4}\right)\pi\right)\right|=\left|\sin\left(\rho\pi\right)\right|\gg\left|\rho\right|.\label{eq:cosLowerBound}
	\end{equation}
The $ t $th zero $j_{1,t}$ of $ J_1 $ satisfies \[j_{1,t} = \left(t+\frac{1}{4}\right)\pi + O(1/t)\]
(see, e.g., \cite[formula (9.5.12)]{AS}), so that
\begin{equation}
\label{eq:distance_to_bessel_zero}
|2\pi T - j_{1,t}|=\pi |\rho| + O(1/T).
\end{equation}
In particular, if $ 2\pi T $ is bounded away from $ j_{1,t} $, then \eqref{eq:distance_to_bessel_zero} yields $ \rho \gg 1 $, so that (\ref{eq:J_1(T)}) and (\ref{eq:cosLowerBound}) give $J_{1}\left(2\pi T\right)^{2}\gg T^{-1}$,
which together with (\ref{eq:VargeJ1}) yields
	\[
	\text{Var}\left(\Dpc_{n;s}\right)\gg T^{-3}.
	\]

Given $ \delta > 0  $, we consider two cases, whether $|\rho| \ge T^{-\delta/2}$ or $|\rho|<T^{-\delta/2}$, aiming  at proving \eqref{eq:var>>1/N^A T^3+del} with the same $ \delta$. If $|\rho| \ge T^{-\delta/2},$ then by \eqref{eq:J_1(T)} and \eqref{eq:cosLowerBound} it follows that $J_{1}\left(2\pi T\right)^{2}\gg T^{-1-\delta}$ so that
	\eqref{eq:VargeJ1} gives
	\begin{equation}
	\label{eq:lower_bound_easy_case}
	\text{Var}\left(\Dpc_{n;s}\right)\gg T^{-3-\delta},
	\end{equation}
	stronger than \eqref{eq:var>>1/N^A T^3+del} with $ A>0 $ arbitrary. Assume otherwise that $|\rho|<T^{-\delta/2}$, and observe that all odd numbers $m\in S$ are expressible as
	\begin{equation}
	\label{eq:m_sum_of_squares}
	m=a^{2}+\left(2k+1-a\right)^{2}
	\end{equation}
	for some $k\ge0$ and $1\le a\le2k+1$.
	Consider all tuples of the form
	\begin{equation}
	\label{eq:special_tuples}
	\left(\lambda^{1},\dots,\lambda^{2k+1}\right)=\left(\stackrel{a\,\text{times}}{\overbrace{\lambda,\dots,\lambda}},\stackrel{2k+1-a\,\text{times}}{\overbrace{i\lambda,\dots,i\lambda}}\right).
	\end{equation}
	The number of such tuples is precisely $N_{n},$ and they satisfy
	\begin{equation}
	\label{eq:special_tuples_sum_norm}
	\left\Vert\lambda^{1}+\dots+\lambda^{2k+1}\right\Vert=\sqrt{nm}.
	\end{equation}
 	By the inequality $\alpha^{2}+\beta^{2}\ge\frac{\left(\alpha+\beta\right)^{2}}{2}$ applied to \eqref{eq:m_sum_of_squares},
	we get that
	$
	m\ge\frac{\left(2k+1\right)^{2}}{2}\ge2k^{2}
	$
	so that
	\begin{equation}
	\label{eq:k_to_m_ineq}
		k\le\sqrt{m/2}.
	\end{equation}
	By the positivity of all the terms in \eqref{eq:VarFormula}, we can bound $ 	\text{Var}\left(\Dpc_{n;s}\right) $ from below by restricting the inner summation in \eqref{eq:VarFormula} to tuples of the form \eqref{eq:special_tuples}. This together with \eqref{eq:special_tuples_sum_norm} and \eqref{eq:k_to_m_ineq} (note that $ a_k  \gg m^{-3/4}$) gives the lower bound
	\begin{equation}
	\text{Var}\left(\Dpc_{n;s}\right)\gg\frac{1}{T^{2}}\sum_{\begin{subarray}{c}
		m\in S\\
		m\,\text{odd}
		\end{subarray}}\frac{1}{N_{n}^{\sqrt{2m}}}\frac{J_{1}\left(2\pi T\sqrt{m}\right)^{2}}{m^{7/4}}.\label{eq:VarLowerBoundS}
	\end{equation}
	
	Let $K>1$ be a sufficiently large parameter to be chosen later, and restrict the summation
	in \eqref{eq:VarLowerBoundS} to primes $p\in P_{K}$ in \eqref{eq:PK primes def} (these are the primes
	$p\equiv1\,\left(\text{mod}\,4\right)$ which are less or equal to $ K $). Then
	\begin{align}
	\text{Var}\left(\Dpc_{n;s}\right) & \gg_{K}\frac{1}{N_{n}^{\sqrt{2K}}T^{2}}\sum_{p\in P_K}
 J_{1}\left(2\pi T\sqrt{p}\right)^{2}.\label{eq:VarLowerBoundNearBesselZeros}
	\end{align}
	By (\ref{eq:BesselAysmp}), we have
	\begin{align}
	J_{1}\left(2\pi T\sqrt{p}\right) & =\pi^{-1}T^{-1/2}p^{-1/4}\cos\left(\left(2T\sqrt{p}-\frac{3}{4}\right)\pi\right)+O\left(T^{-3/2}\right).\label{eq:BesselAsymp2k+1}
	\end{align}
	We write
	\[
	2T\sqrt{p}-\frac{1}{4}=l+\eta
	\]
	where $l=l\left(T,p\right)\in\mathbb{Z}$ and $\left|\eta\right|\le1/2$.
	Then by (\ref{eq:TApprox}),
	\begin{align}
	\left|\cos\left(\left(2T\sqrt{p}-\frac{3}{4}\right)\pi\right)\right| & =\left|\sin\left(\eta\pi\right)\right|\gg\left|\eta\right|=\left|2T\sqrt{p}-l-\frac{1}{4}\right|=\left|\left(t+\frac{1}{4}+\rho\right)\sqrt{p}-l-\frac{1}{4}\right|\nonumber\\
	& \gg\left|\left(4t+1\right)\sqrt{p}-\left(4l+1\right)\right|-4\left|\rho\right|\sqrt{p}.\label{eq:cos_lower_bound_sqrtp}
	\end{align}
	By Lemma \ref{lem:Diop approx roots}, there exists $ p_0 \in P_K $ such that
	\begin{equation}
	\label{eq:schmidt_sqrtp}
	\left|\left(4t+1\right)\sqrt{p_0}-\left(4l+1\right)\right|\gg_{K,\epsilon}t^{-2\log K/K-\epsilon}.
	\end{equation}
	Since $|\rho|<T^{-\delta/2},$ by choosing $K=K\left(\delta\right)$
	sufficiently large so that $ 2\log K/K < \delta / 4$
(keeping in mind that $t = 2T + O(1)$), we conclude upon substituting the bound \eqref{eq:schmidt_sqrtp} in \eqref{eq:cos_lower_bound_sqrtp} that
	\[
	\left|\cos\left(\left(2T\sqrt{p_0}-\frac{3}{4}\right)\pi\right)\right|\gg_{\delta}T^{-\delta/4},
	\]
	which by \eqref{eq:BesselAsymp2k+1} implies
	\begin{equation*}
	J_{1}\left(2\pi T\sqrt{p_0}\right)^{2}\gg_{\delta}T^{-1-\delta/2}.
	\end{equation*}
	This, together with \eqref{eq:VarLowerBoundNearBesselZeros} gives
	\begin{equation}
	\label{eq:lower_bound_hard_case}
		\text{Var}\left(\Dpc_{n;s}\right)  \gg_{\delta }\frac{1}{N_{n}^{\sqrt{2K}}T^{2}}
	J_{1}\left(2\pi T\sqrt{p_0}\right)^{2}\gg_{\delta} \frac{1}{N_n^{\sqrt{2K}} T^{3+\delta/2}}.
	\end{equation}
To summarize, the bounds \eqref{eq:lower_bound_easy_case} and \eqref{eq:lower_bound_hard_case} imply that, in either case,  \eqref{eq:var>>1/N^A T^3+del} holds with $A=$$\sqrt{2K}$, which is the statement
of Theorem \ref{thm:var def ARW low bnd}.
\end{proof}

\section{Spatial defect distribution: proof of Theorem \ref{thm:var upper bnd}}
\label{sec:Bourgain de-randomization}

Recall that Theorem \ref{thm:var upper bnd} follows at once from its more explicit variant, Theorem \ref{thm:derand expl}, whose proof is the ultimate goal of this section.

\subsection{Proof of Theorem \ref{thm:derand expl}}

The following proposition is seemingly weaker, or less general, compared to Theorem \ref{thm:derand expl},
as it only allows for radii $s=\frac{R}{\sqrt{n}}$ with $R\rightarrow\infty$ growing {\em slowly}, instead of a uniform statement for {\em all} $s>R/\sqrt{n}$ as in
\eqref{eq:double lim var vanish}.
However, we will be able to infer the more general result, using the elegant Integral-Geometric Sandwich in Proposition \ref{prop:sandwich IG}
below, inspired to high extent by its counterpart introduced by Nazarov-Sodin ~\cite[Lemma $1$]{SodinSPB} for the sake of counting the number of nodal components (see also
~\cite[Lemma $3.7$]{SaWi} and ~\cite[Lemma $1$]{NSNodal}). It seems a priori {\em counter-intuitive} that it is ``easier" to first establish the spatial defect variance vanishing for smaller radii
than bigger ones.
Our explanation of the said {\em surprise} is that the asymptotic Gaussianity w.r.t. the spatial variable holds at Planck scale only
(or logarithmically above it ~\cite{SartoriLog}), rather than at {\em all} scales above it.

\begin{proposition}[Planck scale spatial defect distribution]
\label{prop:derand Planck scale}
Let $S''\subseteq S$ be any sequence of energy levels satisfying the assumptions of Theorem \ref{thm:derand expl}. Then for every $\epsilon>0$
there exists $R_{0}=R_{0}(\epsilon)>0$ sufficiently large so that for all $R>R_{0}$ there exists a number $n_{0}=n_{0}(R,\epsilon)$ sufficiently large so that for all $n>n_{0}$, the inequality
\begin{equation*}
\var_{\Tb^{2}}(Y_{f_{n},R/\sqrt{n}}) < \epsilon
\end{equation*}
holds uniformly for $f_{n}\in \Bc_{n}$.
Equivalently,
\begin{equation}
\label{eq:var vanish weak}
\lim\limits_{R\rightarrow\infty}\limsup\limits_{\substack{n\rightarrow\infty \\ n\in S''}}
\sup\limits_{f_{n}\in\Bc_{n}}\var_{\Tb^{2}}(Y_{f_{n},R/\sqrt{n}}) = 0.
\end{equation}

\end{proposition}

The following proposition asserts the aforementioned Integral Geometric Sandwich; unlike the original inequality,
it contains an error term. Recall that
the local (normalized) defect of a function an eigenfunction $f_{n}$ as in \eqref{eq:fn toral Laplace eig} restricted to a radius-$s$ ball around $x\in\Tb^{2}$ is given by \eqref{eq:Yfns def spac def}.

\begin{proposition}[Integral Geometric Sandwich]
\label{prop:sandwich IG}
For every $f_{n}$ of the form \eqref{eq:fn toral Laplace eig}, and $0<r_{1}<r_{2}$, the asymptotic estimate
\begin{equation}
\label{eq:sandwich IG}
Y_{f_{n},r_{2}}(x) = \frac{1}{\pi r_{2}^{2}}\int\limits_{B_{x}(r_{2})} Y_{f_{n},r_{1}}(y) dy + O\left(\frac{r_{1}}{r_{2}}   \right)
\end{equation}
holds, with constant associated to the $`O'$-notation absolute.

\end{proposition}

\begin{proof}[Proof of Theorem \ref{thm:derand expl} assuming propositions \ref{prop:derand Planck scale}-\ref{prop:sandwich IG}.]

Let $\epsilon>0$ be given. First, we apply Proposition \ref{prop:derand Planck scale} to obtain a number $R_{0}=R_{0}(\epsilon)$ so that
for all $R>R_{0}$ there exists a number $n_{0}=n_{0}(R,\epsilon)$ so that for $n>n_{0}$ with $n\in S''$, one has
\begin{equation}
\label{eq:Plack scale derand app}
\var_{\Tb^{2}}\left(Y_{f_{n},R/\sqrt{n}}\right) < \frac{\epsilon^{2}}{4},
\end{equation}
uniformly for all $f_{n}\in \Bc_{n}$. We define
\begin{equation}
\label{eq:R0 red def}
R=R(\epsilon):=(R_{0}+1)^{2},
\end{equation}
and claim that with this choice of $R$, the conclusion of Theorem \ref{thm:derand expl} holds, where the corresponding
$n_{0}=n_{0}(R_{0}+1,\epsilon)$, depending on $\epsilon$ only, is the one we received as the output from the application above of Proposition \ref{prop:derand Planck scale}.
For this particular choice of the parameters, the inequality \eqref{eq:Plack scale derand app} reads
\begin{equation}
\label{eq:var R0'+1 vanish}
\var_{\Tb^{2}}\left(Y_{f_{n},(R_{0}+1)/\sqrt{n}}\right) < \frac{\epsilon^{2}}{4},
\end{equation}
valid for all $n\in S''$, $n>n_{0}$ and $f_{n}\in\Bc_{n}$. To validate our claim we are to prove that for all $n>n_{0}$ with $n\in S''$,
the inequality
\begin{equation}
\label{eq:var<eps sufficient}
\var_{\Tb^{2}}\left(Y_{f_{n},s}\right) < \epsilon
\end{equation}
holds for all $s>\frac{R}{\sqrt{n}}$.

Now, we invoke the Integral Geometric Sandwich of Proposition \ref{prop:sandwich IG}, with $r_{2}=s >R/\sqrt{n}$ and
\begin{equation}
\label{eq:r1<r2/R0'}
r_{1}=\frac{R_{0}+1}{\sqrt{n}}<\frac{r_{2}}{R_{0}+1},
\end{equation}
by \eqref{eq:R0 red def}. Hence \eqref{eq:sandwich IG} reads
\begin{equation}
\label{eq:apply sand big rad}
\begin{split}
Y_{f_{n},s}(x) &= Y_{f_{n},r_{2}}(x)
=  \frac{1}{\pi s^{2}}\int\limits_{B_{x}(s)} Y_{f_{n},(R_{0}+1)/\sqrt{n}}(y) dy + O\left(\frac{r_{1}}{s} \right)
\\&=\frac{1}{\pi s^{2}}\int\limits_{B_{x}(s)} Y_{f_{n},(R_{0}+1)/\sqrt{n}}(y) dy + O\left(\frac{1}{R_{0}} \right),
\end{split}
\end{equation}
thanks to \eqref{eq:r1<r2/R0'}. We assume that $R_{0}$ is sufficiently large so that the error term on the r.h.s. of \eqref{eq:apply sand big rad} is $O\left(\frac{1}{R_{0}}\right)<\frac{\epsilon}{2}$, take the absolute value of both sides of \eqref{eq:apply sand big rad},
and apply the triangle inequality
to conclude that
\begin{equation}
\label{eq:sand app concrete}
|Y_{f_{n},s}(x)| \le \frac{1}{\pi s^{2}}\int\limits_{B_{x}(s)} \left| Y_{f_{n},(R_{0}+1)/\sqrt{n}}(y)\right| dy +
\frac{\epsilon}{2}.
\end{equation}

We then integrate both sides of \eqref{eq:sand app concrete} w.r.t. $x\in \Tb^{2}$ to yield
\begin{equation*}
\int\limits_{\Tb^{2}}|Y_{f_{n},s}(x)| dx \le \int\limits_{\Tb^{2}}\left| Y_{f_{n},(R_{0}+1)/\sqrt{n}}(y)\right|dy + \frac{\epsilon}{2},
\end{equation*}
and invoke \eqref{eq:var R0'+1 vanish} together with Cauchy-Schwarz inequality, that gives (recalling that the
spatial expectation vanishes identically, see \eqref{eq:spatial exp vanish})
\begin{equation}
\label{eq:L1 norm def<eps}
\int\limits_{\Tb^{2}}|Y_{f_{n},s}(x)| dx \le \frac{\epsilon}{2} + \frac{\epsilon}{2} = \epsilon.
\end{equation}
Finally, the inequality
\eqref{eq:L1 norm def<eps} certainly implies \eqref{eq:var<eps sufficient},
since $|Y_{f_{n},s}(x)| \le 1$ (again, upon recalling \eqref{eq:spatial exp vanish}),
which, as it was mentioned above, is sufficient to infer the statement of Theorem \ref{thm:derand expl}.

\end{proof}

\subsection{Integral Geometric Sandwich: Proof of Proposition \ref{prop:sandwich IG}}

\begin{proof}
We start with the integral on the r.h.s. of \eqref{eq:sandwich IG}, and use the definition \eqref{eq:Yfns def spac def} to write
\begin{equation}
\label{eq:int def ball double int}
\frac{1}{\pi r_{2}^{2}}\int\limits_{B_{x}(r_{2})} Y_{f_{n},r_{1}}(y) dy = \frac{1}{\pi r_{2}^{2}}\int\limits_{B_{x}(r_{2})}
\frac{1}{\pi r_{1}^{2}}\int\limits_{B_{y}(r_{1})}H(f_{n}(z))dzdy,
\end{equation}
and aim at reversing the order of the integrals on the r.h.s. of \eqref{eq:int def ball double int}.
We have
\begin{equation}
\label{eq:sand inv int}
\frac{1}{\pi r_{2}^{2}}\int\limits_{B_{x}(r_{2})} Y_{f_{n},r_{1}}(y) dy =
\frac{1}{\pi r_{2}^{2}}\int\limits_{B_{x}(r_{2}+r_{1})}H(f_{n}(z)) \cdot \frac{1}{\pi r_{1}^{2}}\vol \left(B_{z}(r_{1})\cap B_{x}(r_{2})\right)  dz.
\end{equation}

Now, upon denoting $$V_{x,z}(r_{2},r_{1}):=
\frac{1}{\pi r_{1}^{2}}\cdot\vol \left(B_{z}(r_{1})\cap B_{x}(r_{2})\right),$$ the equality \eqref{eq:sand inv int} reads
\begin{equation}
\label{eq:sand int den V}
\frac{1}{\pi r_{2}^{2}}\int\limits_{B_{x}(r_{2})} Y_{f_{n},r_{1}}(y) dy = \frac{1}{\pi r_{2}^{2}}\int\limits_{B_{x}(r_{2}+r_{1})}H(f_{n}(z)) \cdot V_{x,z}(r_{2},r_{1})dz,
\end{equation}
and we notice that
\begin{equation}
\label{eq:|V|<=1}
0\le V_{\cdot,z}(\cdot,r_{1})\le \frac{1}{\pi r_{1}^{2}} \vol(B_{z}(r_{1})) = 1,
\end{equation}
and, in addition, if $z\in B_x(r_{2}-r_{1})$, then $V(z)=1$. We then separate the range of integration in \eqref{eq:sand int den V}
into $B_x(r_{2}-r_{1})$ and its complement to write
\begin{equation}
\label{eq:sand int sep in out}
\begin{split}
\frac{1}{\pi r_{2}^{2}}\int\limits_{B_{x}(r_{2})} Y_{f_{n},r_{1}}(y) dy &= \frac{1}{\pi r_{2}^{2}}\int\limits_{B_{x}(r_{2}-r_{1})}H(f_{n}(z))dz
+ O\left(\frac{1}{\pi r_{2}^{2}}\vol(B_{x}(r_{2}+r_{1})\setminus B_{x}(r_{2}-r_{1}))\right)
\\&=\frac{1}{\pi r_{2}^{2}}\int\limits_{B_{x}(r_{2})}H(f_{n}(z))dz
+ O\left(\frac{1}{\pi r_{2}^{2}}\vol(B_{x}(r_{2}+r_{1})\setminus B_{x}(r_{2}-r_{1}))\right)
\\&= Y_{f_{n},r_{2}}(x)
+ O\left(\frac{1}{\pi r_{2}^{2}}\vol(B_{x}(r_{2}+r_{1})\setminus B_{x}(r_{2}-r_{1}))\right)
\end{split}
\end{equation}
thanks to \eqref{eq:|V|<=1}, $|H(\cdot)|\le 1$, and the definition \eqref{eq:Yfns def spac def} of $Y_{f_{n},r_{2}}(x)$.
Now the statement \eqref{eq:sandwich IG} of Proposition \ref{prop:sandwich IG} finally follows from substituting the estimate
\begin{equation*}
\frac{1}{\pi r_{2}^{2}}\vol(B_{x}(r_{2}+r_{1})\setminus B_{x}(r_{2}-r_{1})) = O\left( \frac{r_{2}r_{1}}{r_{2}^{2}}  \right) =
O\left( \frac{r_{1}}{r_{2}}  \right)
\end{equation*}
into \eqref{eq:sand int sep in out}.

\end{proof}

\subsection{Auxiliary results towards the proof of Proposition \ref{prop:derand Planck scale}}
\label{sec:aux res derand}

We denote Berry's random monochromatic isotropic waves $g:\R^{2}\rightarrow\R$ defined on a probability space
$(\Omega,\Sigma,\prob)$, i.e. for $\omega\in\Omega$ the corresponding sample function $g(\cdot)=g_{\omega}(\cdot)$ are distributed as a
centred Gaussian random field uniquely determined via Kolmogorov's Theorem by its covariance function
\begin{equation}
\label{eq:Berry RWM covar}
r_{g}(|x-y|):=\E[g(x)\cdot g(y)] = J_{0}(|x-y|),
\end{equation}
where $J_{0}$ is the Bessel $J$ function of order $0$.
Proposition \ref{prop:meas pres map} immediately below asserts that locally, the functions $f_{n}\in\Bc_{n}$, appropriately scaled,
converge to $g(\cdot)$ around a random spatial variable on the torus, understood as random fields. It is the heart of Bourgain's de-randomization
method, originally in ~\cite{Bo2013},
and is a restatement of what turned out to be the key technical propositions in ~\cite{BW}, in the precise form used in that manuscript.
To state this result, given a function $f_{n}\in\Bc_{n}$, we introduce
the function $F_{x;R}(y):[-1,1]^{2}\rightarrow\R$ to be
\begin{equation}
\label{eq:F scal def}
F_{x;R}(y) = f_{n}\left(x+\frac{R}{\sqrt{n}}y \right),
\end{equation}
and think of $F_{x;R}(\cdot)$ as a random field, as $x\in \Tb^{2}$ varies randomly uniformly on the torus.
In what follows we will obtain a sequence of random fields $g^{n}:\R^{2}\rightarrow\R$, that will converge
in suitable sense to $g$, and we will denote their scaled version
\begin{equation*}
g^{n}_{\omega;R}(\cdot):= g^{n}_{\omega}(\cdot R),
\end{equation*}
that will be compared to the scaled version of $g$
\begin{equation}
\label{eq:g scal def}
g_{\omega;R}(\cdot):= g_{\omega}(\cdot R).
\end{equation}

\begin{proposition}[{~\cite[Propositions 3.2-3.3]{BW}}]
\label{prop:meas pres map}
Let $S''\subseteq S$ be a sequence of energy levels satisfying the assumptions of Theorem \ref{thm:derand expl}.
Then there exists a sequence of Gaussian stationary random fields $\{g^{n}\}_{n\in S''}$, converging in law to $g$ as $n\rightarrow \infty$, with the following property.
For every $R>0$, $\epsilon>0$ and $\eta>0$, there exists $n_{0}=n_{0}(R;\eta,\epsilon)$ sufficiently large so that for all
$n\in S''$ with $n>n_{0}$ and $f_{n}\in \Bc_{n}$, there exists an event
$\Omega'=\Omega'(n;f_{n},R;\eta,\epsilon)\subseteq \Omega$ of high probability $\prob(\Omega')>1-\epsilon$
and a measure preserving map $\tau:\Omega'\rightarrow\Tb^{2}$ so that $\meas(\tau(\Omega'))>1-\epsilon$, and for all $\omega\in \Omega'$, one has
\begin{equation}
\label{eq:C1 estimate gn}
\|g^{n}_{\omega;R}- F_{\tau(\omega);R}\|_{C^{1}([-1,1]^{2})} < \eta.
\end{equation}
\end{proposition}

Since, as it was mentioned above, Proposition \ref{prop:meas pres map} was proved\footnote{In ~\cite{BW} a more general situation was
considered, when the equidistribution assumption on the lattice points was lifted, whence the limit random field was varying,
depending on their angular distribution, rather than the sole Berry's random waves limit field $g(\cdot)$.} in ~\cite{BW},
there is no need to reprove it in this manuscript. Once the reduction to the Gaussian random field was performed within Proposition
\ref{prop:meas pres map}, replacing $g^{n}(\cdot)$ with Berry's $g(\cdot)$ in \eqref{eq:C1 estimate gn} is completely standard.
That is, it is possible to
couple $g^{n}(\cdot)$ with $g(\cdot)$ so that $\|g^{n}_{\omega;R}- g_{\omega;R}\|_{C^{1}([-1,1]^{2})}$ is arbitrarily small
for $n$ sufficiently large, see e.g. ~\cite[Lemma 4]{SodinSPB}. Together with \eqref{eq:C1 estimate gn} and the triangle inequality it yields the
following corollary.

\begin{corollary}
\label{cor:meas pres map}
Let $S''\subseteq S$ be a sequence of energy levels satisfying the assumptions of Theorem \ref{thm:derand expl}. Then for every $R>0$, $\epsilon>0$ and $\eta>0$, there exists $n_{0}=n_{0}(R;\eta,\epsilon)$ sufficiently large so that for all
$n\in S''$ with $n>n_{0}$ and $f_{n}\in \Bc_{n}$, there exists an event
$\Omega'=\Omega'(n;f_{n},R;\eta,\epsilon)\subseteq \Omega$ of high probability $\prob(\Omega')>1-\epsilon$
and a measure preserving map $\tau:\Omega'\rightarrow\Tb^{2}$ so that $\meas(\tau(\Omega'))>1-\epsilon$, and for all $\omega\in \Omega'$, one has
\begin{equation}
\label{eq:C1 estimate}
\|g_{\omega;R}- F_{\tau(\omega);R}\|_{C^{1}([-1,1]^{2})} < \eta.
\end{equation}
\end{corollary}

Alternatively to working with $g(\cdot)$, one could, in principle, work directly with
$g^{n}(\cdot)$, by proving an analogue of Lemma \ref{lem:var def RWM bnd} below, applicable for $g^{n}(\cdot)$ with $n$ large, a direction we abandon.
Corollary \ref{cor:meas pres map} naturally gives rise to the comparison to the defect variance of the random waves $g(\cdot)$.
Note that, for our purposes of comparing
the defect of the toral eigenfunctions to that of the random $g_{R}$, the $C^{1}$-estimate in \eqref{eq:C1 estimate} is too strong, and we could easily settle for an $L^{\infty}$-estimate.
Recall that $H(\cdot)$ is the sign function \eqref{eq:H Heaviside}, and let
\begin{equation}
\label{eq:XR defect RWM}
X_{R}=X_{\omega,R} := \frac{1}{\pi R^{2}}\int\limits_{B(R)}H(g(x))dx
\end{equation}
be the (random) defect of $g(\cdot)$ restricted to the
ball $B(R)\subseteq \R^{2}$. It is obvious that the expectation $\E[X_{R}]=0$ vanishes, whereas the following easy, most likely sub-optimal, result asserts
that so does its variance, asymptotically as $R\rightarrow\infty$.

\begin{lemma}
\label{lem:var def RWM bnd}
As $R\rightarrow\infty$, the defect variance of $g(\cdot)$ restricted to $B(R)$ is vanishing:
\begin{equation}
\label{eq:var def RWM bnd}
\var(X_{R}) = O\left(\frac{1}{R^{1/2}} \right).
\end{equation}
\end{lemma}

\begin{proof}
We use the definition \eqref{eq:XR defect RWM} of the defect, and invert the integration order to write
\begin{equation}
\label{eq:def arcsin}
\var(X_{R}) = \frac{2}{\pi^{3}R^{4}}\int\limits_{B(R)\times B(R)} \arcsin (J_{0}(|x-y|))dxdy,
\end{equation}
where we reused the well-known identity \eqref{eq:arcsineFormula}.
Now, for each $x\in B(R)$ fixed we separate the range of integration in \eqref{eq:def arcsin} into $|x-y|<1$ and $|x-y|>1$ (say),
so that
\begin{equation}
\label{eq:var loc glob sep}
\begin{split}
\var(X_{R}) &= \frac{2}{\pi^{3}R^{4}}\cdot \left( \int\limits_{\substack{x,y\in B(R)\\|x-y|<1}} \arcsin (J_{0}(|x-y|))dxdy +
\int\limits_{\substack{x,y\in B(R)\\|x-y|>1}} \arcsin (J_{0}(|x-y|))dxdy  \right)\\& =:
 \frac{2}{\pi^{3}R^{4}} \cdot (I_{1}+I_{2}).
 \end{split}
\end{equation}

We bound the contribution of the former range trivially as
\begin{equation}
\label{eq:var loc bnd}
|I_{1}| = \left|\int\limits_{\substack{x,y\in B(R)\\|x-y|<1}} \arcsin (J_{0}(|x-y|))dxdy\right| = O(R^{2}),
\end{equation}
whereas we use the standard asymptotics ~\cite[formula (9.2.1)]{AS} for the Bessel $J_{0}$ function for $|x-y|>1$:
\begin{equation*}
|\arcsin(J_{0}(t))|\ll |J_{0}(t)| \ll \frac{1}{\sqrt{t}}
\end{equation*}
to bound the contribution of the latter range as
\begin{equation}
\label{eq:var I2 glob bnd}
\begin{split}
I_{2} &=\int\limits_{\substack{x,y\in B(R)\\|x-y|>1}} \arcsin (J_{0}(|x-y|))dxdy \ll
\int\limits_{B(R)}dx \int\limits_{y\in B(R): \: |x-y|>1}\frac{dy}{|x-y|^{1/2}}
\\&\le R^{2}\int\limits_{1}^{R} \frac{tdt}{\sqrt{t}}  \ll R^{7/2}.
\end{split}
\end{equation}
The statement of Lemma \ref{lem:var def RWM bnd} finally follows upon substituting \eqref{eq:var loc bnd} and \eqref{eq:var I2 glob bnd} into \eqref{eq:var loc glob sep}.

\end{proof}

We will require the following notion, inspired by ~\cite{SodinSPB,NSNodal}, that will allow us to control the defect stability under small $L^{\infty}$-perturbations.

\begin{definition}[Stable event]
For $R>0$, $\eta>0$ and $\delta>0$ we let the ``$(R;\eta,\delta)$-unstable" event $\Omega_{1}(R;\eta,\delta)\subseteq \Omega$ be defined
as
\begin{equation}
\label{eq:Omega1 stable def}
\Omega_{1}(R;\eta,\delta):= \left\{ \omega\in\Omega:\: \frac{1}{\pi R^{2}}\cdot\meas\{x\in B(R):\: |g_{\omega}(x)| <\eta\}>\delta \right\}
\end{equation}
the event that the proportion of $x\in B(R)$ so that $|g(x)|$ is small, is not negligible.
\end{definition}

\begin{lemma}[Stability estimate]
\label{lem:stability}
For every $\delta,\epsilon>0$, there exists an $\eta>0$ sufficiently small, so that
for every $R>0$,
\begin{equation*}
\prob (\Omega_{1}(R;\eta,\delta))<\epsilon.
\end{equation*}
\end{lemma}

\begin{proof}

Let $\Ac_{R;\eta}\subseteq B(R)$ be the (random) measure
$$\Ac_{R;\eta}:=\meas\{x\in B(R):\: |g(x)|<\eta\}$$ of the set $g^{-1}([-\eta,\eta])\cap B(R) \subseteq \R^{2}$. Clearly,
\begin{equation}
\label{eq:Aeta=int char}
\Ac_{R;\eta} = \int\limits_{B(R)} \chi_{[-\eta,\eta]}(g(x))dx,
\end{equation}
where $\chi_{[-\eta,\eta]}$ is the characteristic function of the interval $[-\eta,\eta]\subseteq\R$.
Since, for every $x\in\R^{2}$, $g(x)$ is a standard Gaussian random variable, taking the expectation of both sides of
\eqref{eq:Aeta=int char} easily yields
\begin{equation}
\label{eq:exp ind=O(eta R^2)}
\E[\Ac_{R;\eta}] = O(\eta R^{2}),
\end{equation}
with the constant involved in the `O'-notation absolute. Now, we have
\begin{equation*}
\Omega_{1}(R;\eta,\delta) = \left\{\omega\in\Omega:\: \frac{1}{\pi R^{2}}\cdot \Ac_{R;\eta} > \delta  \right\},
\end{equation*}
and, in light of \eqref{eq:exp ind=O(eta R^2)}, the conclusion of Lemma \ref{lem:stability} follows from Markov's inequality.

\end{proof}

After all the preparatory results of \S\ref{sec:aux res derand}, we are finally in a position to prove the principal de-randomization result.

\subsection{Spatial defect distribution: Proof of Proposition \ref{prop:derand Planck scale} via Bourgain's de-randomization}
%\label{sec:spac def dist derand}
We start with the following elementary lemma in probability theory, that is a criterion for the variance vanishing of {\em bounded} random variables, whose proof is thereupon conveniently omitted.

\begin{lemma}
\label{lem:var vanishing bnded rvs}
Let $\{X_{k}\}_{k\ge 1}$ be a sequence of random variables $X_{k}:\Omega\rightarrow\R$ on a probability space $(\Omega,\Sigma,\prob)$ satisfying $|X|\le 1$ a.s. and
$\E[X_{k}] = 0$ for every $k\ge 1$. Then we have
$\var(X_{k})\rightarrow 0$ as $k\rightarrow\infty$, if and only if for
every $\delta>0$, the probability $\prob(|X_{k}|>\delta)\rightarrow 0$
vanishes as $k\rightarrow \infty$.

\end{lemma}

\begin{proof}[Proof of Proposition \ref{prop:derand Planck scale}]

We are going to use Lemma \ref{lem:var vanishing bnded rvs} as a criterion for the variance vanishing, upon both exploiting the
defect variance for Berry's random waves (Lemma \ref{lem:var def RWM bnd}), and also when proving the same for the toral eigenfunctions;
note that the prescribed rate \eqref{eq:var def RWM bnd} is ``lost" during this process for the latter.
Let $\epsilon,\delta>0$ be given. First, we invoke Lemma \ref{lem:stability} on $\delta/4$ in place of $\delta$, and
$\epsilon/2$ in place of $\epsilon$, to obtain a number $\eta=\eta(\epsilon/2,\delta/4)$ sufficiently small so that
for all $R>0$,
\begin{equation}
\label{eq:prob(Omega1)<eps/2}
\prob(\Omega_{1}(R,\eta,\delta/4)) < \epsilon/2.
\end{equation}

Next, we apply on Lemma \ref{lem:var def RWM bnd} (along with the ``only if" statement of Lemma \ref{lem:var vanishing bnded rvs}),
to obtain a number
$R_{0}=R_{0}(\delta/2,\epsilon/4)$ sufficiently large, so that for all $R>R_{0}$, we have
\begin{equation*}
\prob\left\{ |X_{R}| >\frac{\delta}{2}\right\} < \frac{\epsilon}{4}.
\end{equation*}
Let $\Omega_{2}\subseteq \Omega$ be the corresponding event, i.e.
\begin{equation}
\label{eq:Omega2 def}
\Omega_{2}=\Omega_{2}(R;\delta/2) := \left\{ |X_{R}| >\frac{\delta}{2}\right\},
\end{equation}
of probability
\begin{equation}
\label{eq:prob(Omega2)<eps/4}
\prob(\Omega_{2})<\frac{\epsilon}{4}.
\end{equation}

Finally, we apply Corollary \ref{cor:meas pres map} to obtain a number $n_{0}=n_{0}(R;\eta,\epsilon/4)$, so that for all $n>n_0$
and $f_{n}\in \Bc_{n}$ there
exists an event $\Omega'=\Omega'(n;f_{n},R;\eta,\epsilon/4)$ of probability
\begin{equation}
\label{eq:prob(Omega')>1-eps/4}
\prob(\Omega')>1-\epsilon/4,
\end{equation}
and a measure preserving map
$\tau:\Omega'\rightarrow\Tb^{2}$ so that
\begin{equation}
\label{eq:C1 estimate apply}
\|g_{\omega;R}- F_{\tau(\omega);R}\|_{C^{1}([-1,1]^{2})} < \eta,
\end{equation}
where $g_{\omega;R}$ are the (scaled) Berry's random waves \eqref{eq:g scal def}, and $F_{\tau(\omega);R}$
is the scaled version of the given $f_{n}\in\Bc_{n}$, defined in \eqref{eq:F scal def}.

Recall that $X_{\omega;R}$ is the defect \eqref{eq:XR defect RWM} of
Berry's random waves restricted to $B(R)$.  In light of \eqref{eq:C1
  estimate apply}, for $y\in [-1,1]^{2}$ we
have $$H(g_{\omega;R}(y))=H(F_{\tau(\omega);R}(y)),$$ unless
$|g_{\omega;R}(y)|<\eta$. Hence, by the definition \eqref{eq:Omega1
  stable def} of the unstable event $\Omega_{1}$, it is clear
(the magnitude of change in the sign
function is at most $2$, and the measure of the set of $x$ for which
$|g_{\omega}(x)| < \eta  $ is at most $\delta/4$) that
for all $\omega\in \Omega'\setminus\Omega_{1}$, one has
\begin{equation}
\label{eq:defect stab}
|X_{\omega,R}-Y_{f_{n},R/\sqrt{n}}(\tau(\omega))|<2\cdot\frac{\delta}{4}=\frac{\delta}{2}.
\end{equation}
Now, by the definition of $\Omega_{2}$, for every $\omega\notin\Omega_{2}$ one has
\begin{equation}
\label{eq:Xomega,r<delt/2}
|X_{\omega;R}|<\frac{\delta}{2}.
\end{equation}
Hence \eqref{eq:Xomega,r<delt/2} together with \eqref{eq:Omega2 def}
imply that for all $\omega\in
\Omega'':=(\Omega'\setminus\Omega_{1})\setminus\Omega_{2}$, one has
\begin{equation*}
|Y_{f_{n},R/\sqrt{n}}(\tau(\omega))| \le |X_{\omega,R}| + \frac{\delta}{2} < \delta.
\end{equation*}
Equivalently,
\begin{equation}
\label{eq:|Y|<delt}
|Y_{f_{n},R/\sqrt{n}}(x)| < \delta
\end{equation}
for all $x\in \tau(\Omega'')$ of measure
\begin{equation}
\label{eq:meas(tau(Omega''))}
\meas(\tau(\Omega'')) \ge \prob(\Omega')-\prob(\Omega_{1})-\prob(\Omega_{2}) > (1-\frac{\epsilon}{4})-\frac{\epsilon}{2}-
\frac{\epsilon}{4} = 1-\epsilon,
\end{equation}
thanks to \eqref{eq:prob(Omega1)<eps/2}, \eqref{eq:prob(Omega2)<eps/4} and \eqref{eq:prob(Omega')>1-eps/4},
and the measure preserving property of $\tau$.
Finally, \eqref{eq:|Y|<delt}, \eqref{eq:meas(tau(Omega''))}, and the ``if" direction of Lemma \ref{lem:var vanishing bnded rvs}
allow us to deduce the conclusion of Proposition \ref{prop:derand Planck scale}.

\end{proof}

\section{Eigenfunctions with non-vanishing defect variance: proof of Theorem \ref{thm:nonflat large def}}
\label{sec:def non vanish}

\subsection{Large negative defect on hexagonal lattices}

We begin by constructing a completely flat Laplace eigenfunction $g$
on a certain {\em hexagonal} torus $T$, such that the {\em total}
defect of $g$ is non-vanishing. In what follows it will be convenient to identify $\R^2$ with $\C$.

Define $L:= \Z[ 1 + i/\sqrt{3}, 2i/\sqrt{3}]$, and let
$T := \C/ L$. Further, let $\hat{L} \subset \C \simeq \R^2$ denote the dual lattice
to $L$, generated by the sixth roots of unity (or just by $\{1, e(1/6)\}$, where
$e(z) := e^{2 \pi i z} $). The Laplace eigenvalues on $T$ are then given by
$4 \pi^{2} |v|^{2}$ for $v \in \hat{L}$.
Let $v_{1}, \ldots, v_{6} \in \R^2$ denote the six elements in
$\hat{L}$ with length one, and for $x \in \R^2$ define
$
f(x) =
\sum_{i=1}^{6}  e( v_{i} \cdot x);
$
$f$ is then well defined on $T$ (as well as totally flat), and is a
Laplace eigenfunction on $T$, with eigenvalue $4 \pi^{2}$.

Further, let $w_{1},w_{2},w_{3} \in \R^{2}$
denote elements corresponding to the three third roots of unity.
Using that $e(t)+e(-t) = 2 \cos(t)$,
and pairing off antipodal points (i.e. $v_{i} = -v_{j}$)
define the completely flat function
\begin{equation}
\label{eq:g hex torus def}
g(x) :=
\sum_{i=1}^{3}  \cos( 2 \pi w_{i} \cdot x)
=f(x) / 2.
\end{equation}
Further, $g_{m}(x) := g(mx)$ is
a Laplace eigenfunction on $T$ with eigenvalue $4 \pi^{2} m^{2}$ (also
completely flat if $m$ is chosen to be a prime
that is inert in $\Z[e^{2 \pi i/3}]$, and the following proposition asserts
that the total defect of $g$ does not vanish.

\begin{proposition}
\label{prop:negative-hex-defect}
We have
\begin{equation}
\label{eq:tot def g=c<0}
c := \int\limits_{T} H(g(y)) \, dy < 0.
\end{equation}
Further, for any $x \in T$, and $s>0$
$$
\frac{1}{\pi s^{2}}\int_{B_{x}(s)} H(g_{m}(y)) \, dy
= c \cdot \frac{\sqrt{3}}{2} + O( 1/( ms))
$$
\end{proposition}

  \begin{figure}[h]
    \centering
\includegraphics[width=8cm]{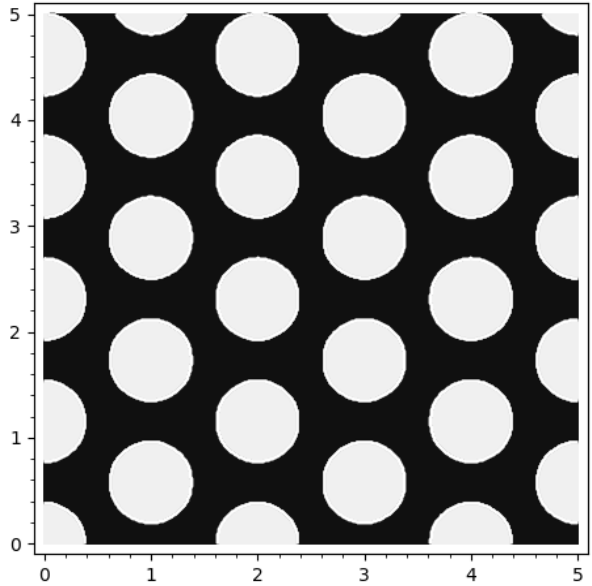}
\caption{
%Center of first ball
%  along $x$-axis: $2$. Center of first ball along $y$-axis:
%  $2 /\sqrt{3} \simeq 1.155$.
White regions denotes $g(x_{1},x_{2})>0$,
  and black denotes $g(x_{1},x_{2})<0$. Despite appearances, the white
  regions are
{\em not} circles.  }
\label{fig:hexpic}
  \end{figure}

A plot of $H(g(x_{1},x_{2}))$ is shown in Figure~\ref{fig:hexpic}.
Since $g$ is invariant under translation by $ L$, unless the
integral over the fundamental domain of $L$ is exactly zero, we will
get growth, of order $R^{2}$ in either the positive or the negative
direction, when integrating over squares, say centred at $(R/2,R/2)$
and with sides length $R$ growing. The numerics in
Table~\ref{tab:integral} indicates that there is negative
growth. These numerics can be made rigorous by bounding the gradient
from above: this way we can ensure that the function does not change
sign in most small disks. The following lemma, whose proof is
obvious, introduces a stability notion, related to the one in
section \ref{sec:aux res derand}.

\begin{table}
  \centering
  \begin{tabular}{c| c| c|}
R &  $\int_{S(R)} H(g(x))  \,dx$ & $(1/R^{2}) \cdot\int_{S(R)}
                                            H(g(x)) \,dx $ \\
\hline
5 & -5.10561833230128 & -0.204224733292051  \\
%10 & -20.7813895675727 & -0.207813895675727  \\
15 & -43.5759827038652 & -0.193671034239401  \\
%20 & -86.1374614923547 & -0.215343653730887  \\
25 & -116.854534058787 & -0.186967254494059  \\
%30 & -166.288744604681 & -0.184765271782979  \\
35 & -247.264843494327 & -0.201848851832104  \\
\hline
  \end{tabular}
  \caption{Integral values.  Here $S(R) \subset \R^2$ is the square
    $[0,R]\times[0,R]$.}
\label{tab:integral}
\end{table}

\begin{lemma}
\label{lem:stab}
For the function $g$ in \eqref{eq:g hex torus def} define
\begin{equation}
\label{eq:M grad def}
M := \max_{x \in T}  | \nabla g (x)|,
\end{equation}
and let $D_{x}(r)$ denote a closed disk of radius $r>0$ centred at $x$.
Then $M \le 2 \pi \cdot 3$, and
$$
  \min_{y \in D_{x}(r)} |g(y)|
  \ge |g(x)| - r \cdot M.
$$
\end{lemma}

\begin{proof}[Proof of Proposition~\ref{prop:negative-hex-defect}]

Recall that the  lattice ${L}$ is spanned by
$u_{1}=(1,1/\sqrt{3})$ and $u_{2} = (0,2/\sqrt{3})$. The rhombus
spanned by $u_{1},u_{2}$ is a fundamental domain of ${L}$, as
well as a fundamental domain for $T$.
As it is more convenient to tile with rectangles rather than with
rhombi we will prefer to evaluate the signed area on a rectangular fundamental
domain, and show that the defect integral over the rectangle
$\mathcal{R}$, having corners at
$(0,0), (1,0), (0, 2/ \sqrt{3}), (1, 2 / \sqrt{3})$,
easily seen to be a fundamental domain of $T$, is non-zero.

\vspace{2mm}

For some integer $N>0$ we tile $\mathcal{R}$ by $N^{2}$
rectangles (modulo $\mathcal{R}$) centred at
$$h_{j,k} = \left( \frac{j}{N},  \frac{k}{N}\cdot \frac{2}{\sqrt{3}} \right)$$ for $0 \le j,k < N$;
each such rectangle can be covered with a disk of radius
$r = \sqrt{7/12}/N$. If the inequality $|g(h_{j,k})| > 12 \pi r > r\cdot M$,
with $M$ as in
\eqref{eq:M grad def} is satisfied (using a factor of two safety
margin), the corresponding rectangle centred at $h_{j,k}$
is said to be ``stable'', whence $g(\cdot)$ has constant sign on the
whole rectangle by
Lemma \ref{lem:stab}; otherwise it is said to be ``unstable''. Depending on
the sign of $g(h_{j,k})$, we call the corresponding stable rectangle ``positively stable''
or ``negatively stable''.

For $N=80$ one finds $2099$ positively stable rectangles, $3299$ negatively stable, and $1002$
unstable ones. As $3299-2099 =1200> 1002$, we conclude that the
defect \eqref{eq:tot def g=c<0} is nonzero (and in fact negative).
Both assertions of
Proposition~\ref{prop:negative-hex-defect} now follow: the first assertion
follows from the presented numerical calculation, whereas the second
one is an immediate
consequence of the first assertion upon tiling $B_{x}(s)$ with
$\pi (ms)^{2}/(2/\sqrt{3}) + O(ms)$ copies of fundamental domains
associated with the
lattice $\frac{1}{m} L$ (note that the boundary of $B_{x}(s)$ can be
covered with $O(ms)$ tiles.)
One can obtain more precise estimates on $c$ in
\eqref{eq:tot def g=c<0}, by increasing $N$, and thus
decreasing the mesh size: for example, for $N=500$, the
corresponding counts are respectively $96639, 147207$, and $6154$.
\end{proof}

\subsection{Defect stability w.r.t. perturbations of $g$}
\label{sec:stability}
For later use we show that a small perturbation of $g$ only changes
the defect by a small amount. For convenience we work in the rescaled
region where the eigenvalues are normalized to $4 \pi^{2}$, hence we
should consider the defect over balls of radius $R$ (or squares of
sides $R$) with $R$ growing.
We start by showing that simultaneous vanishing of both $g$ and its
gradient $\nabla g$ is impossible.

\begin{lemma}
\label{lem:g nonsing}
Let $Z_{1} := \{ x \in T : g(x) =
   0 \}$ and let $Z_{2} := \{ x \in T : \nabla g (x) = (0,0) \}$.
   Then $Z_{1} \cap Z_{2} = \emptyset$.
\end{lemma}
\begin{proof}
The linear map $\R^3 \to \R^{2} $, given by
$(a_{1},a_{2},a_{3}) \to \sum_{i=1}^{3} a_{i} w_{i}$ with $w_{i}$ as
in \eqref{eq:g hex torus def}, clearly has
full range, hence a one dimensional kernel, spanned by
  $(1,1,1)$. In particular, if $\sum_{i=1}^{3} a_{i} w_{i} = 0$, then
  $a_{1}=a_{2}=a_{3} = C$ for some $C$. Therefore, $\nabla g (x) = 0$
  implies that $\cos( 2\pi w_{1} \cdot x) = \cos( 2\pi w_{2} \cdot x) = \cos( 2\pi
  w_{3} \cdot x) = C$
  for some $C$.  Further, $g(x)=0$ implies that
  $0= \sum_{i=1}^{3} \cos(2\pi  w_{i} \cdot x) = 3C$, and thus $C=0$ for any point
  where $g$ and $\nabla g$ both vanish.  In particular, we find that
  $2\pi w_{i} \cdot x = \pm \pi/2 + 2 \pi k_{i}$ for $k_{i} \in \Z$.  On the
  other hand, as $\sum_{i=1}^{3} w_{i} = 0$, we find, on multiplying
  by $2/\pi$ that
$$
0 \equiv  \pm 1 + \pm 1 + \pm 1 \mod 4
$$
which is impossible since the right hand side is odd no matter what
signs are chosen.
\end{proof}

In light of Lemma \ref{lem:g nonsing} and the compactness of $T$, it follows
that the gradient of $g$ is uniformly bounded below on the zero set of $g(\cdot)$:

\begin{corollary}
\label{cor:bnd grad zero set}
There exist $C>0$ such that $|\nabla g (x) | \geq C$ for all $x \in
Z_{1}=g^{-1}(0)$.
\end{corollary}

% Given a nice subset $X \subset \R^2$ and a continuous function $f \in
% C(X)$  (or $f \in C(\R^{2})$), define
% $$
% D(f,X) =
% \int_{X} \operatorname{sgn}(f(x)) \, dx,
% $$
% where
% $$
% \operatorname{sgn}(t) =
% \begin{cases}
% 1 & \text{if $t > 0$,}  \\
% -1 & \text{if $t < 0$,}  \\
% 0 & \text{if $t = 0$.}
% \end{cases}
% $$

It is now straightforward to prove stability of the defect of $g$ w.r.t. perturbations. Given
$R \ge 1$ and a continuous function $f \in C(\R^{2})$, define
$$
Y_{f,R}(x) :=
\frac{1}{\pi R^{2}}
\int_{B_{x}(R)} H(f(y)) \, dy,
$$

\begin{lemma}
\label{lem:stability2}
Let $g$ be the function \eqref{eq:g hex torus def}, and $R\ge 1$. Then for all $\epsilon>0$ sufficiently small, if
$f \in C(\R^{2})$ is such that $|g(y)-f(y)| < \epsilon$
holds for all $y \in B_{x}(R)$, one has
$$
Y_{f,R}(x) =
Y_{g,R}(x) + O(\epsilon).
$$
\end{lemma}
\begin{proof}
It is sufficient to show that the measure of the set
$$\{ x \in T : |g(x)| \leq \epsilon \}$$ is $O(\epsilon)$, for all
sufficiently small $\epsilon$, as we can then tile $B_{x}(R)$ with
$\sim R^{2}$ copies of the fundamental domain.  Now, there exist some
open neighborhood of $Z_{1}=g^{-1}(0)$, outside of which $|g(x)|$ is
uniformly bounded away from zero (say, using compactness of the
closed complement). In other words, if $|g(x)|$ is small then we
must have $d(x,Z_{1})$ small, where $d(x,Z_{1})$ denotes the
distance between $x$ and the zero set $Z_{1}$.
Further, all  $x$ for which $d(x,Z_{1})$ is  sufficently small is
contained in some small
tubular neighbourhood of $Z_{1}$.
The lower bound on the gradient of Corollary \ref{cor:bnd grad zero set} implies
that $|g(x)| \gg d(x,Z_{1}) + O(d(x,Z_{1})^{2} )$, and hence the
measure of the set of $x$ for which $|g(x)| < \epsilon$ is
$\ll \epsilon $.
\end{proof}

%\section{The regular torus $\T = \R^2/\Z^2$}
\subsection{Approximating $g$ on the standard torus $\Tb^{2} = \R^2/\Z^2$: proof of Theorem \ref{thm:nonflat large def}}
\label{sec:regular-torus}

We next show that a perturbed variant of the hexagonal lattice
construction can be translated to  the square torus.
We begin by showing that the set of Gaussian integers, scaled to have
norm one, can very well approximate third roots of unity.

\begin{proposition}
\label{prop:inf int Pell}
The Pell equation
\begin{equation}
\label{eq:Pell eq}
b^{2}-3a^{2}=1
%3a^{2}-b^{2}=-1
\end{equation}
admits infinitely many solutions. Further, let
\begin{equation}
\label{eq:S''' def}
S'''=\{n=a^{2}+b^{2}\}
\end{equation}
be the infinite sequence of
integers of the form $a^{2}+b^{2}$ with $(a,b)$ as in \eqref{eq:Pell eq}, and for $n\in S'''$ we define the Gaussian integers
$z_{1}=z_{n,1},z_{2}=z_{n,2},z_{3}=z_{n,3}$ as
\begin{equation}
\label{eq:z123 def}
z_{1} := -a + bi, \quad z_{2} := -a-bi, \quad z_{3} := 2a+i.
\end{equation}
Then, as $n\rightarrow\infty$ along $S'''$, we have
\begin{equation}
\label{eq:z1,z2,z3 prop}
z_{1}/|z_{1}| = e^{2 \pi i/3} + O\left(n^{-1/2}\right) \quad
z_{2}/|z_{2}| = e^{-2 \pi i/3}  + O\left(n^{-1/2}\right), \quad
z_{3}/|z_{3}| = 1 + O\left(n^{-1/2}\right).
\end{equation}
\end{proposition}
\begin{proof}
Since the Pell equation $b^{2}-3a^{2} = 1$ has the solution
$a=1,b=2$, it has infinitely many integer solutions.
Moreover, we find that $|z_{1}|^{2} = |z_{2}|^{2} = |z_{3}|^{2} = 4a^{2}+1$, and
\begin{equation}
\label{eq:w123 approx wtild123}
\frac{z_{1}}{|z_{1}|} = \frac{-1+i \sqrt{3}}{2} + O(\frac{1}{a}),
\quad
\frac{z_{2}}{|z_{2}|} = \frac{-1-i \sqrt{3}}{2} + O(\frac{1}{a}),
\quad
\frac{z_{3}}{|z_{3}|} = 1 + O(\frac{1}{a}).
\end{equation}
Thus, taking $n = a^{2}+b^{2} = 4a^{2} + 1$ we have $1/a =
O(n^{-1/2})$, and the proof of Proposition \ref{prop:inf int Pell} is concluded.
\end{proof}

\begin{proof}[Proof of Theorem \ref{thm:nonflat large def}]

We claim that the statement of Theorem \ref{thm:nonflat large def} holds, with $S'''$ prescribed by \eqref{eq:S''' def},
satisfying, in particular, the statement \eqref{eq:z1,z2,z3 prop} of Proposition \ref{prop:inf int Pell}.
To construct eigenfunctions on $\T=\R^2/\Z^{2}$ having large defect it
is convenient to rescale $\T$ so that the eigenvalue equals $4\pi^{2}$, and
correspondingly the torus must be rescaled so that the fundamental
domain is a square with sides $n^{1/2}$ (where $\lambda = 4\pi^{2} n$ denotes
the unscaled eigenvalue.)
Given $n = a^{2}+b^{2} \in S'''$ with $b^{2}-3a^{2} = 1$ define the unit
vectors $\widetilde{w_{i}}:= \frac{z_{i}}{|z_{i}|}\in\R^{2}$, $i=1,2,3$, with $z_{i}$ as in \eqref{eq:z123 def},
and the Laplace eigenfunction $G$, on the re-scaled torus
$\R^2/(\sqrt{n} \Z^{2})$, by
$$
G(x) :=
\sum\limits_{i=1}^{3} \cos( 2 \pi \widetilde{w_{i}} \cdot x)
$$

A simple calculation shows that $G$ is a Laplace eigenfunction, with
eigenvalue $4 \pi^{2}$, and that, with $w_{i}$ as in \eqref{eq:g hex torus def},
the asymptotic approximation \eqref{eq:w123 approx wtild123} reads
$$
|w_{i}-\widetilde{w_{i}}| = O(1/a) = O(1/n^{1/2}).
$$
Hence, for any $x \in \R^2$, we have
$$
|g(x)-G(x)| \ll |x|/n^{1/2}.
$$
In particular, for $|x|= o(n^{1/2})$, we have $G(x) = g(x) + o(1)$,
and thus, if $R = o(n^{1/2})$ grows with $n$ we find,
thanks to Lemma~\ref{lem:stability2}, that
$$
Y_{G,R}(x) =
Y_{g,R}(x) + o(1) = C + o(1)
$$
for $C := c \cdot \sqrt{3}/2 < 0$.  In the macroscopic regime, i.e.
when $R$ is of size $n^{1/2}$, we similarly find that for
$|x|  \ll \epsilon n^{1/2}$,
$$
Y_{G,R}(x) =
Y_{g,R}(x) + O(\epsilon) = C + O(\epsilon).
$$

Thus, if for $n \in S'''$ we construct $G$ as described above and define
$f_{n}(x) := G(\sqrt{n} x)$, we obtain an eigenfunction on $\Tb^{2}$, with
eigenvalue $4\pi^2 n$, and find that the defect integral over
$B_{x}(s)$ (keeping in mind that $s = R/\sqrt{n}$ when we undo the
scaling) is bounded away from zero for $|x| < \epsilon$; hence the variance
is bounded from below, and the proof is concluded.

\end{proof}

\end{document}